\documentclass[10pt,journal]{IEEEtran}
\usepackage{amsmath,amsfonts}
\usepackage{algorithmic}
\usepackage{algorithm}
\usepackage{array}
\usepackage[caption=false,font=normalsize,labelfont=sf,textfont=sf]{subfig}
\usepackage{textcomp}
\usepackage{stfloats}
\usepackage{url}
\usepackage{verbatim}
\usepackage{graphicx}
\usepackage{cite}
\usepackage{xcolor}
\graphicspath{{figure/}}
\usepackage{multirow}
\usepackage{rotating}
\usepackage{graphicx}
\usepackage{amsthm}
\usepackage{amssymb}
\newtheorem{theorem}{Theorem}
\newtheorem{lemma}{Lemma}
\newtheorem{example}{Example}
\newtheorem{definition}{Definition}

\newtheorem{remark}{Remark}

\DeclareMathOperator{\diag}{diag}
\DeclareMathOperator{\argmin}{argmin}
\DeclareMathOperator{\hd}{HD}
	\definecolor{gray(x11gray)}{rgb}{0.9, 0.9, 0.9}
\begin{document}

\title{Fast Successive-Cancellation Decoding of $2\times2$ Kernel Non-Binary Polar Codes: Identification, Decoding and Simplification
}

%\author{IEEE Publication Technology,~\IEEEmembership{Staff,~IEEE,}
        % <-this % stops a space
%\thanks{This paper was produced by the IEEE Publication Technology Group. They are in Piscataway, NJ.}% <-this % stops a space
%\thanks{Manuscript received April 19, 2021; revised August 16, 2021.}
%}
\author{Ali Farsiabi, Hamid Ebrahimzad, Masoud Ardakani,~\IEEEmembership{Senior Member,~IEEE}, Chuandong Li}% Masoud Ardakani,~\IEEEmembership{Senior Member,~IEEE,} Chuandong Li}
{%Shell \MakeLowercase{\textit{et al.}}: A Sample Article Using IEEEtran.cls for IEEE Journals}


\maketitle

\begin{abstract}
Non-binary polar codes (NBPCs) decoded by successive cancellation (SC) algorithm have remarkable bit-error-rate performance compared to the binary polar codes (BPCs). Due to the serial nature, SC decoding suffers from large latency. The latency issue in BPCs has been the topic of extensive research and it has been notably resolved by the introduction of fast SC-based decoders. However, the vast majority of research on NBPCs is devoted to issues concerning design and efficient implementation. In this paper, we propose fast SC decoding for NBPCs constructed based on $2\times 2$ kernels. In particular, we identify various non-binary special nodes in the SC decoding tree of NBPCs and propose their fast decoding. This way, we avoid traversing the full decoding tree and significantly reduce the decoding delay compared to symbol-by-symbol SC decoding. We also propose a simplified NBPC structure that facilitates the procedure of non-binary fast SC decoding. Using our proposed fast non-binary decoder, we observed an improvement of up to $95\%$ in latency concerning the original SC decoding. This is while our proposed fast SC decoder for NBPCs incurs no error-rate loss. 
\end{abstract}


\section{Introduction}
\IEEEPARstart{P}{olar} codes, invented by Arıkan, represent a channel forward-error-correction (FEC) scheme that can asymptotically achieve the capacity of binary-input memoryless channels \cite{arikan}. The explicit construction of polar codes and their low-complexity successive cancellation (SC) decoding algorithm has captured significant interest in both academia and industry. In particular, binary polar codes (BPCs) are adopted in the 5G standard as the FEC scheme in the control channel of the enhanced mobile broadband (eMBB) use case \cite{5G}. Although SC decoding of long-length polar codes provides a low-complexity capacity-achieving solution, its sequential bit-by-bit decoding schedule leads to high decoding latency, which limits its adoption in low-latency applications such as optical fiber communications or ultra-reliable low-latency communication (URLLC) scheme of 5G. Therefore, the design of fast-SC-based decoding algorithms is of significant practical interest.

The decoding latency problem of BPCs has been the topic of extensive research in recent years and various schemes have been developed to address this issue \cite{sarkis2014, yazdi, sarkis2018, hanif, condo, ercan, hashemilist2016, hashemilist2017,  hashemiflex, ardakani, LiLetter, SR, SRlist}. The main underlying idea behind these schemes is to perform some operations in parallel to avoid the serial nature of SC decoding. In particular, in these schemes, using the binary tree representation of SC decoding of polar codes, certain special nodes are recognized whose decoding can be performed in parallel. For example, identification and parallel decoding of Rate-0 and Rate-1 nodes was proposed in \cite{yazdi}. Similarly, in \cite{sarkis2014}, fast decoding of single parity check (SPC) and repetition (REP) nodes was introduced. On the basis of the works in \cite{yazdi, sarkis2014}, five new nodes (Type-I, Type-II, Type-III, TypeIV, Type-V) were identified in \cite{hanif}, and their fast decoding was formulated. To reduce the SC-decoding latency even further, generalized REP (G-REP) and generalized parity check (G-PC) nodes and their decoding were proposed in \cite{condo}. Recently, sequence repetition (SR) node, which is a generalization of the G-REP node, along with its decoding algorithm was introduced in \cite{SR}. Moreover, the fast decoding of the aforementioned special nodes were also extended to SC list (SCL) \cite{talvardy} decoding \cite{hashemilist2016, hashemilist2017,  SRlist,ardakani}.

As a promising FEC scheme, non-binary polar codes (NBPCs) can asymptotically achieve the capacity of discrete memoryless channels with arbitrary $q$-ary alphabets \cite{sasoglu, mori_confRS, mori_RSpolar, park_qary}. Importantly, due to symbol-level operations, the SC decoding latency in NBPCs is saved compared to bit-level operations of BPCs. The vast majority of research on non-binary polar codes is devoted to the issues concerning design and efficient implementations \cite{trifinovLett, savin, emspolar}. In \cite{sasoglu, mori_confRS, mori_RSpolar, 2x2NBkernel, gf3gf5}, NBPCs with different non-binary kernels were investigated and it has been shown that even with a $2 \times 2$ kernel, they achieve impressive bit error rates (BERs) compared to their binary counterparts. To the best of our knowledge, the only work regarding the fast decoding of NBPCs is \cite{rate_1NBP} in which Rate-1 nodes were used to simplify the SCL decoding. 

Motivated by the widespread demand for low-latency capacity-achieving FEC schemes, in this paper, we propose the fast SC decoding of non-binary polar codes constructed by $2\times 2$ kernels. In fact, to the best of our knowledge, this is the first work in which the identification of a variety of non-binary special nodes along with their decoding algorithms is proposed. We start by formulating the NBPC decoding algorithm in a form appropriate for identifying special nodes. We also discuss a hardware-friendly version of the algorithm that can be implemented using a limited number of quantization bits. We then identify various non-binary special nodes and propose low-complexity decoding algorithms for each of them. We also propose a simplified NBPC structure in which the elimination of certain permutations and multiplications in the proposed structure reduces the complexity of fast decoding. Our simulation shows that while the proposed fast SC decoder for NBPCs incurs almost no error-correction performance loss, the latency can be reduced by as much as $95\%$ percent.

The rest of the paper is organized as follows: Section II provides some background information on non-binary polar codes. This is followed by the presentation of the decoding algorithm of NBPCs in Section III. We then present our proposed fast decoding of NBPCs based on non-binary special nodes in Section IV. In Section V, we propose a simplified structure for NBPCs that is suitable for fast decoding.
In Section VI, the latency of our proposed fast decoder is analyzed. To
evaluate the accuracy of the proposed fast decoder, we present the simulation results in Section VII. We end the paper with some concluding remarks in Section VIII.

\section{Background}
\subsection{Notations}
In this paper, matrices, and vectors are denoted by boldfaced upper case and lower case letters, respectively. Moreover, all the vectors are assumed to be column vectors.
\subsection{Non-binary polar codes}\label{sec_nbpc}
In this paper, a non-binary polar code  (NBPC) with symbol-length $N=2^n$, message length $K$, and rate $R=K/N$ is defined over the non-binary Galois Field $\mathbb{GF}(q)$ where\footnote{The choice of $q$ as a power of $2$ may be of practical interest for hardware implementation. Nevertheless, our proposed method in this paper is applicable to any non-binary field.} $q=2^p$ with $p>1$. The non-binary elements over $\mathbb{GF}(q)$ can be expressed as $\alpha^{-\infty}=0, \alpha^0, \alpha^1,\dots,\alpha^{q-2}$ where $\alpha$ is the root of a primitive polynomial $f(x)=a_0+a_1x+a_2x^2+\dots+a_px^p$, $a_i\in \mathbb{GF}(2)$. Each of the elements in $\mathbb{GF}(q)$ can be represented by a binary vector of size $p$. Therefore, a sequence of information bits, $\boldsymbol{b}=[b_{0}, \dots, b_{K_b-1}]^T$, with length $K_b=pK$ can be converted to a non-binary sequence, $\mathbf{m} = [m_{0}, \dots, m_{K-1}]^T$, with length $K$ over $\mathbb{GF}(q)$. In the NBPC encoder, the elements of $\mathbf{m}$ are placed at $K$ information positions of \emph{input sequence} $\mathbf{u}=[u_0, \dots, u_{N-1}]^T$ and the remaining positions are filled with 0-symbols known as \emph{frozen} symbols.  The vector $\mathbf{u}$ is encoded into the non-binary polar codeword $\mathbf{c}=[c_0, \dots, c_{N-1}]^T$ using the following equation
\begin{equation}
\mathbf{c}^T=\mathbf{u}^T\mathbf{G}^{\bigotimes n}_2,\label{eq:polarencfix}
\end{equation}
where $\bigotimes$ denotes the Kronecker power and $\mathbf{G}^{\bigotimes n}_2$ is the generator matrix\footnote{$\mathbf{G}^{\bigotimes n}_2$ is the generator matrix of NBPCs in which the kernel coefficients are fixed in all the stages of polarization. In Section \ref{sec:Gnbpc}, we will describe the generalized NBPC structure in which the kernel coefficients can vary across the polarization units.} of NBPC. The kernel $\mathbf{G}_2$ that is used in this paper is the extension of $2\times 2$ Arıkan kernel to non-binary GF, which is written as  
\begin{equation}\label{eq:G2nb}
\mathbf{G}_2 = \begin{bmatrix} \mu & 0 \\ 
\gamma &  \delta \end{bmatrix}
\end{equation}
where $\mu, \gamma, \delta \in \mathbb{GF}(q)  \backslash 0$, referred to as kernel coefficients, are non-zero elements in $\mathbb{GF}(q)$. The set of information and frozen positions, denoted by $\mathcal{I}$ and $\mathcal{I}^c$, are typically optimized for a specific signal-to-noise ratio (SNR) and are determined using Monte-Carlo (MC) simulation \cite{nbplowlatency,2x2NBkernel}.

At the transmitter, depending on the application requirements, the non-binary codeword $\mathbf{c}$ can either directly be mapped to $p$-ary constellation symbols (such as $2^p$-QAM) or, alternatively, it can be converted to a binary sequence $\mathbf{t}=[t_0, \dots, t_{N_b-1}]^T$ with $N_b=p\times N$ elements and, likewise the binary FECs, be mapped to different symbols of an arbitrary modulation scheme (BPSK, M-QAM, ...). %In this paper, without loss of binary phase shift keying (BPSK) modulation that maps $\{0,1\}$ to $\{-1,1\}$ is considered. At the transmitter, the non-binary codeword $\mathbf{c}$ is converted to a binary sequence $\mathbf{t}=[t_0, \dots, t_{N_b-1}]^T$ with $N_b=p\times N$ elements and modulated using BPSK. Transmission takes place over the additive white Gaussian noise (AWGN) channel and the received sequence $\mathbf{y}$ at the channel output is obtained as $\mathbf{y}=(1-2\mathbf{t}) + \mathbf{z}$. The noise vector $\mathbf{z}$ consists of $N_b$ independent and identically distributed (i.i.d.)
%Gaussian random variables with zero mean and variance $\sigma^2$. Based on the channel output, the bit log-likelihood ratios (LLRs) can be obtained as
%\begin{equation}
%\eta^{\text{ch}}_k = \frac{2y_k}{\sigma^2}, \: \: \: \: \: 0\leq k \leq N_b-1.\label{eq:bitLLRs}
%\end{equation}
%As shown in the next section, the bit LLRs are used to obtain symbol LLRs corresponding to the non-binary codeword-symbols.  
\section{SC Decoding of NBPC}\label{sec_emssc}
\begin{figure}[!t]
\centering
\includegraphics[width=1.5in]{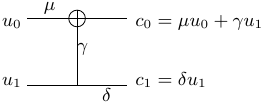}
\caption{Non-binary $2 \times 2$ kernel with $\mu, \gamma, \delta \in \mathbb{GF}(q)  \backslash 0$.}
\label{G2kernel}
\end{figure}
Figure \ref{G2kernel} shows the schematic of kernel $\mathbf{G}_2$ which is equivalent to the transformation $\mathcal{T}: (c_0, c_1) = (\mu u_0 + \gamma u_1, \delta u_1)$ in which the addition and multiplication are performed over $\mathbb{GF}(q)$. In fact, an NBPC with length $N=2^n$ is obtained by $n$ stages of recursive application of $\mathcal{T}$. As $\mathcal{T}$ is composed of a check node (CN) and a variable node (VN), there are $nN/2$ CNs/VNs in an NBPC of length $N$. In this section, the update rule of the CNs/VNs in the decoding process is described.
% \ref{sec_emssc}.
%\begin{figure}[!t]
%\centering
%\includegraphics[width=1.9in]{G2nb.eps}
%\caption{Non-binary kernel $\mathbf{G}_2$, with $\gamma$ a non-zero element of $\mathbb{GF}(q)$.}
%\label{G2kernel}
%\end{figure}
\begin{figure}[!t]
\centering
\includegraphics[width=3in]{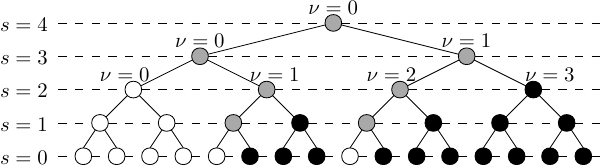}
\caption{Binary-tree representation of a non-binary polar code with $N=16$ and $K=8$. The black and white leaf nodes are information and frozen symbols, respectively.}
\label{bintree}
\end{figure}

Inspired by the state-of-the-art decoding algorithms from non-binary LDPC (NBLDPC) codes literature, we use the LLR-domain NB decoding algorithm of \cite{ems_LDPC} along with successive cancellation (SC) decoding \cite{arikan} schedule to decode the NBPCs. We refer to this algorithm as LLR-domain NB successive-cancellation (LNBSC).

Successive cancellation decoding can be understood using a binary-tree representation of the polar code. In Fig. \ref{bintree}, a binary tree for an NBPC $(16, 8)$ is shown where $s$ denotes the level in the decoding tree and $0 \leq \nu \leq 2^{n-s}$ is the node index from left to right. $N_s=2^s$ is used to denote the length of a node rooted at level $s$. The nodes of the tree can be identified with $(\nu, s)$ pair where each one, except the leaf nodes, has two children: the left child-node $(2\nu, s-1)$ and the right child-node $(2\nu+1, s-1)$. For example, $(0, n)$ denotes the root node at the top of the tree with $(0, n-1)$ and $(1, n-1)$ as its children. Also, $(j, 0)$ represents the $j$th leaf node.

The LNBSC algorithm is implemented in the LLR domain. As there are $q=2^p$ possibilities for each codeword symbol, a scalar LLR, such as bit-LLRs in binary polar decoding, is not sufficient to represent the messages during the decoding process. As such, for each codeword symbol, an LLR vector is defined to reflect all the $q$ possibilities. In this regard, the vector of channel symbol-LLRs associated to the $i$th codeword symbol $c_i$, $0\leq i \leq N-1$, is defined as $\boldsymbol{\ell}^{\text{ch}}_i=[\ell^{\text{ch}}_{i,0},\dots,\ell^{\text{ch}}_{i,{q-1}}]^T$, where the elements of $\boldsymbol{\ell}^{\text{ch}}_i$ are obtained as
\begin{equation}
\ell^{\text{ch}}_{i,\theta} =\ln\frac{p(c_i = \hat{\theta}|\mathbf{y})}{p(c_i = \theta|\mathbf{y})}, \: \:  \: \: \:\hat{\theta}, \theta \in \mathbb{GF}(q).
\label{symbLLR}
\end{equation}

%The EMSSC algorithm can be implemented in LLR domain where the vector of channel symbol LLRs associated to $i$th codeword symbol $c_i$, $0\leq i \leq N-1$, is denoted by $\boldsymbol{\ell}_i=[\ell_{i,0},\dots,\ell_{i,{q-1}}]^T$. The elements of $\boldsymbol{\ell}_i$ are obtained as
%\begin{equation}
%\ell_{i,\theta} =\ln\frac{p(c_i = \hat{\theta}|\mathbf{y})}{p(c_i = \theta|\mathbf{y})}, \: \:  \: \: \:\hat{\theta}, \theta \in \mathbb{GF}(q).
%\label{symbLLR}
%\end{equation}
The symbol $\hat{\theta}$ denotes the most likely symbol associated with $c_i$. Based on \eqref{symbLLR}, all the elements of $\boldsymbol{\ell}^{\text{ch}}_i$ are non-negative and $\ell^{\text{ch}}_{i,\hat{\theta}}=\ln(1)=0$ meaning that at least one element of $\boldsymbol{\ell}^{\text{ch}}_i$ is equal to zero. The symbol LLRs in \eqref{symbLLR} can either be calculated from received $p$-ary constellation symbols or extracted from bit log-likelihood ratios (LLRs). Regarding the latter case, let the channel bit LLRs corresponding to different bits be denoted by $\eta^{\text{ch}}_k$, $0 \leq k < N_b-1$. Also, suppose the binary representation of $\theta \in \mathbb{GF}(q)$ be denoted as $(\theta(0),\dots,\theta(p-1))$. Then, it can be shown that the channel symbol-LLRs, $\ell^{\text{ch}}_{i,\theta}$, and the channel bit-LLRs, have the following relation,
\begin{equation}
\ell^{\text{ch}}_{i,\theta}=\sum_{j=0}^{p-1}(\theta(j)\oplus \hd(\eta^{\text{ch}}_{ip+j})|\eta^{\text{ch}}_{ip+j}| , \: \: \: \: \: 0\leq i \leq N-1,
\label{bin2symllr}
\end{equation}
where $\oplus$ is XOR operator and $\hd(\cdot)$ makes hard decision on bit-LLRs as
\begin{equation}
{\hd(\eta)} = \begin{cases}
0,&{\text{if}}\ \eta > 0; \\ 
{1,}&{\text{otherwise.}} 
\end{cases}
\end{equation}
%In \eqref{bin2symllr}, it is assumed that every $p$ adjacent bits of the transmitted sequence represent the binary equivalent of a non-binary codeword symbol.
\begin{figure}[!t]
\centering
\includegraphics[width=3.2in]{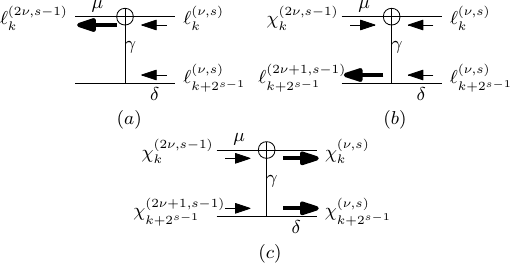}
\caption{Three types of messages that are calculated at node $(
\nu, s)$: (a) CN operation (message toward the left child), (b) VN operation (message toward the right child), and (c) message toward the parent node. }
\label{LR_msg}
\end{figure}

 Decoding the NBPC codes can be viewed as the exchange of messages among the nodes of the binary tree. Let us denote the soft information associated to $i$th symbol of the $(\nu, s)$ node by the length-$q$ vector  $\boldsymbol{\ell}_{i}^{(\nu, s)}=[{\ell}_{i,0}^{(\nu, s)},\dots,{\ell}_{i,q-1}^{(\nu, s)}]^T$, $0\leq i \leq N_s-1$. Each node $(\nu, s)$ receives a $q\times N_s$ soft-information matrix $\bold{L}^{(\nu, s)}=[\boldsymbol{\ell}_0^{(\nu, s)}, \dots, \boldsymbol{\ell}_{N_s-1}^{(\nu, s)}]$  from its parent node and computes the vector of hard-symbol estimates $\boldsymbol{\chi}^{(\nu, s)} = [\chi_0^{(\nu, s)}, \dots, \chi_{N_s-1}^{(\nu, s)}]^T$. The hard-symbols, $\boldsymbol{\chi}^{(\nu, s)}$, are then passed to the parent node, i.e., $(\lfloor \nu/2 \rfloor,s+1)$. Based on these notations, the root node receives $\bold{L}^{(0, n)}=[\boldsymbol{\ell}^{\text{ch}}_0, \dots, \boldsymbol{\ell}^{\text{ch}}_{N-1}]$ from the channel and returns the estimated NBPC codeword $ \boldsymbol{\chi}^{(0, n)}=[\chi_0^{(0, n)}, \dots, \chi_{N-1}^{(0, n)}]^T$, $N=2^n$, as its output.

\begin{remark}
Suppose the $i$th symbol of the node $(\nu,s)$ is denoted by $c^{(\nu,s)}_i$ and its associated soft vector is $\boldsymbol{\ell}_i^{(\nu,s)}$. The LLR vector associated to symbols $\gamma c^{(\nu,s)}_i$ and $c^{(\nu,s)}_i+\beta$, $\gamma, \beta \in \mathbb{GF}(q)$, can be obtained by permuting the elements of $\boldsymbol{\ell}_i^{(\nu,s)}$.
\end{remark}
\begin{example}
Suppose $f(x)=x^2+x+1$ is adopted as the primitive polynomial and $f(\alpha)=0$. The elements of $\mathbb{GF}(4)$ can then be expressed as $\{0,1,\alpha, \alpha^2\}$ with binary representation  $\{(0,0),(0,1),(1,0),(1,1)\}$ and decimal representation $\{0,1,2,3\}$. Assume the LLR vector associated to $c^{(\nu,s)}_i$ is $\boldsymbol{\ell}_i^{(\nu,s)}=[\ell_0, \ell_1,\ell_2,\ell_3]^T$, and we would like to obtain the LLR-vectors $\boldsymbol{\ell}_i^{(\nu,s)^\times}$ and $\boldsymbol{\ell}_i^{(\nu,s)^+}$ corresponding to symbols $\alpha c^{(\nu,s)}_i$ and $c^{(\nu,s)}_i + 1$, respectively. Since we have
\begin{equation*}
\begin{split}
\{0,1,\alpha, \alpha^2\}\xrightarrow{\times \alpha}\{0,\alpha,\alpha^2, 1\},\\
\{0,1,\alpha, \alpha^2\}\xrightarrow{+ 1}\{1,0,\alpha^2, \alpha\},
\end{split}
\end{equation*}
we can write
\begin{equation*}
\boldsymbol{\ell}_i^{(\nu,s)^\times} = [\ell_0, \ell_2, \ell_3, \ell_1]^T,
\boldsymbol{\ell}_i^{(\nu,s)^+} = [\ell_1, \ell_0, \ell_3, \ell_2]^T.
\end{equation*}
Therefore, the LLR vectors can be expressed as $\boldsymbol{\ell}_i^{(\nu,s)^\times}=\boldsymbol{\Pi}^{\times}_{\alpha}\boldsymbol{\ell}_i^{(\nu,s)}$ and $\boldsymbol{\ell}_i^{(\nu,s)^+}=\boldsymbol{\Pi}^+_{1}\boldsymbol{\ell}_i^{(\nu,s)}$, where
\begin{equation*}
\boldsymbol{\Pi}^\times_{\alpha}=\left[ \begin{array}{cccr}
1& 0 & 0& 0\\ 
0 & 0 & 1 & 0\\
0& 0 & 0 & 1 \\
0 & 1 & 0 & 0
\end{array}\right],
\boldsymbol{\Pi}^+_{1}=\left[ \begin{array}{cccr}
0& 1 & 0& 0\\ 
1 & 0 & 0 & 0\\
0& 0 & 0 & 1 \\
0 & 0 & 1 & 0
\end{array}\right].
\end{equation*}
\end{example}
\begin{definition}\label{def:confset}
For a symbol $\phi \in \mathbb{GF}(q)$, the configuration set $\mathcal{S}_{\phi}$ is defined as $\mathcal{S}_{\phi} =\{(\zeta_0, \zeta_1) \in \mathbb{GF}(q)|\zeta_0+ \zeta_1 = \phi \}$.
\end{definition}
\begin{example}
Suppose $f(x)=x^2+x+1$ is adopted as the primitive polynomial of $\mathbb{GF}(4)$ with $f(\alpha)=0$. We can define the following configuration sets:
\begin{equation*}
\begin{split}
\mathcal{S}_{0}=\{(0,0),(1,1),(\alpha,\alpha), (\alpha^2, \alpha^2)\},\\
\mathcal{S}_{1}=\{(0,1),(1,0),(\alpha,\alpha^2), (\alpha^2, \alpha)\},\\
\mathcal{S}_{\alpha}=\{(0,\alpha),(\alpha,0),(1,\alpha^2), (\alpha^2, 1)\},\\
\mathcal{S}_{\alpha^2}=\{(0,\alpha^2),(\alpha^2,0),(1,\alpha), (\alpha, 1)\}.
\end{split}
\end{equation*}
\end{example}
With the exception of leaf nodes, upon receiving $\bold{L}^{(\nu, s)}$, each node in the polar code tree sequentially generates two matrices of soft information for its children, i.e., first $\bold{L}^{(2\nu, s)}$ is generated and after receiving the estimated output from the left child, $\bold{L}^{(2\nu+1, s)}$ is calculated. With respect to the CN operation in Fig. \ref{LR_msg}.a, suppose the vector $\tilde{\boldsymbol{\ell}}_{k+2^{s-1}}^{(\nu, s)}=[\tilde{{\ell}}_{k+2^{s-1},0}^{(\nu, s)}, \dots, \tilde{{\ell}}_{k+2^{s-1},q-1}^{(\nu, s)}]^T$ is defined as
\begin{equation}
\tilde{\boldsymbol{\ell}}_{k+2^{s-1}}^{(\nu, s)}=\boldsymbol{\Pi}^\times_\gamma\boldsymbol{\Pi}^{\times}_\frac{1}{\delta} \boldsymbol{\ell}_{k+2^{s-1}}^{(\nu, s)}, 0\leq k \leq 2^{s-1},
\label{perm_left}
\end{equation}
where $\boldsymbol{\Pi}^\times_\gamma$ and $\boldsymbol{\Pi}^{\times}_\frac{1}{\delta}$ are $q\times q$ permutation matrices corresponding to $\gamma$ and inverse of $\delta$ in $\mathbb{GF}(q)$, respectively. Then, the message toward the left child node can be calculated as
\begin{equation}
{\ell}_{k,\phi}^{(2\nu, s-1)}=-\ln\bigg(\sum_{(\zeta_0, \zeta_1) \in \mathcal{S}_{\phi}}e^{-({\ell}_{k, \zeta_0}^{(\nu, s)}+\tilde{{\ell}}_{k+2^{s-1}, \zeta_1}^{(\nu, s)})}\bigg),
\label{NBcnSPA}
\end{equation}
where $0\leq k \leq 2^{s-1}$ and $\mathcal{S}_\phi$ is the configuration set defined in Definition \ref{def:confset}. Note that because of the coefficient $\mu$, the elements of $\boldsymbol{\ell}_{k}^{(2\nu, s-1)}$ are required to be re-permuted as
\begin{equation}
\boldsymbol{\ell}_{k}^{(2\nu, s-1)} =\boldsymbol{\Pi}^{\times}_\frac{1}{\mu} \boldsymbol{\ell}_{k}^{(2\nu, s-1)}.
\end{equation}

\begin{remark}
A hardware-friendly approximation of Equation \eqref{NBcnSPA} known as extended min-sum (EMS) can also be obtained as
\begin{equation}
{\ell}_{k,\phi}^{(2\nu, s-1)}=\min_{(\zeta_0, \zeta_1) \in \mathcal{S}_{\phi}}({\ell}_{k, \zeta_0}^{(\nu, s)}+\tilde{{\ell}}_{k+2^{s-1}, \zeta_1}^{(\nu, s)}).
\label{NBcn}
\end{equation}
This implementation will be examined in the simulation result section.
\end{remark}

 After receiving the vector of hard symbols from the left child, the VN operation in Fig. \ref{LR_msg}.b is performed. Let the vectors $\breve{\boldsymbol{\ell}}_{k}^{(\nu, s)}=[\breve{{\ell}}_{k,0}^{(\nu, s)},\dots,\breve{{\ell}}_{k,q-1}^{(\nu, s)}]^T$ and $\breve\breve{\boldsymbol{\ell}}_{k+2^{s-1}}^{(\nu, s)}=[\breve\breve{{\ell}}_{{k+2^{s-1}},0}^{(\nu, s)},\dots,\breve\breve{{\ell}}_{{k+2^{s-1}},q-1}^{(\nu, s)}]^T$be defined as
\begin{equation}\label{eq:gperm}
\breve{\boldsymbol{\ell}}_{k}^{(\nu, s)}=\boldsymbol{\Pi}_\frac{1}{\gamma}^{\times}\boldsymbol{\Pi}^+_{\mu\chi_k^{(2\nu,s-1)}}\boldsymbol{\ell}_{k}^{(\nu, s)},
\end{equation}
\begin{equation}
\breve\breve{\boldsymbol{\ell}}_{k+2^{s-1}}^{(\nu, s)}=\boldsymbol{\Pi}_\frac{1}{\delta}^{\times}\boldsymbol{\ell}_{k+2^{s-1}}^{(\nu, s)},
\end{equation}
where $\boldsymbol{\Pi}^+_{\mu\chi_k^{(2\nu,s-1)}}$ is a $q \times q$ permutation matrix associated with the hard-symbol $\mu\chi_k^{(2\nu,s-1)}$ and $\boldsymbol{\Pi}_\frac{1}{\gamma}^{\times}$, $\boldsymbol{\Pi}_\frac{1}{\delta}^{\times}$ are the inverse of permutation matrices $\boldsymbol{\Pi}^\times_\gamma$, $\boldsymbol{\Pi}^\times_\delta$. The soft information toward the right child node can then be computed as
\begin{equation}
\acute{{\ell}}_{k+2^{s-1},\phi}^{(2\nu+1, s-1)}=\breve{{\ell}}_{k,\phi}^{(\nu, s)} + \breve\breve{\ell}_{k+2^{s-1},\phi}^{(\nu, s)},\label{NBvn}
\end{equation}
\begin{equation}
{{\ell}}_{k+2^{s-1},\phi}^{(2\nu+1, s-1)} =\acute{{\ell}}_{k+2^{s-1},\phi}^{(2\nu+1, s-1)}- \min_{\beta \in \mathbb{GF}(q)}(\acute{{\ell}}_{k+2^{s-1},\beta}^{(2\nu+1, s-1)}),
\label{bias_rem}
\end{equation}
for $0\leq k \leq 2^{s-1}$ and $\phi 
\in \mathbb{GF}(q)$.
Following the symbol-LLR definition in \eqref{symbLLR}, Equation \eqref{bias_rem} is used to ensure that at least one element of the vector ${{\ell}}_{k+2^{s-1},\phi}^{(2\nu+1, s-1)}$ is equal to zero. On the other hand, the hard symbols in Fig. \ref{LR_msg}.c toward the parent node can be computed as
\begin{equation}
{\chi_{i}^{(\nu, s)}} = \begin{cases}
\mu\chi_{i}^{(2\nu, s-1)}+\gamma \chi_{i}^{(2\nu+1, s-1)},&{\text{}}\ 0 \leq i < 2^{s-1}; \\ 
{\delta\chi_{i-2^{s-1}}^{(2\nu+1, s-1)},}&{\text{}}\ 2^{s-1} \leq i < 2^{s},
\end{cases}\label{NBupmsg}
\end{equation}
where the addition and multiplication are done over $\mathbb{GF}(q)$.

The decoding process starts from the root node (node $(0,n)$) at the top of the tree by receiving $\bold{L}^{(0, n)}$, i.e., channel symbol-LLRs. Then, each node computes the LLR values toward its children until the leaf nodes receive the soft information from their parents. At the $i$th leaf node (node $(i, 0)$), the $i$th input symbol $u_i$ is estimated as
\begin{equation}
{\hat{u}_i}=\boldsymbol{\chi}^{(i,0)} = \begin{cases}
0,&{\text{if}}\ i \in \mathcal{I}^c;\\ 
{\mathcal{H}(\mathbf{L}^{(i,0)}),}&{\text{otherwise,}} 
\end{cases}
\end{equation}
where $\mathcal{H}(\mathbf{L})$ makes non-binary hard-decision on each column of $\mathbf{L}$ as\footnote{Regarding the leaf nodes, $\mathbf{L}^{(i,0)}$ has only one column, i.e., $\boldsymbol{\ell}^{(i,0)}_0$.}
\begin{equation}
\mathcal{H}(\boldsymbol{\ell})=\underset{\phi \in \mathbb{GF}(q)}{\argmin}(\ell_\phi).
\label{NB_harddec}
\end{equation}
The logic behind the definition of the above non-binary hard decision is that, according to \eqref{symbLLR}, the minimum LLR value in an LLR vector corresponds to the maximum likelihood symbol.
\section{Proposed Fast-LNBSC Decoding Based on Special Nodes}\label{sec:fastdecoding}
We find the following notations and definition adopted from \cite{hanif, ardakani} quite useful for better understanding the fast decoder.

\begin{definition}\label{def_sets}
Suppose $\{\!\{ N_s\}\!\}$ denotes $\{0,1,\dots,N_s-1\}$ where $N_s=2^s$. For a node $(\nu,s)$, the sets $\mathcal{A}_{(\nu,s)}$ and $\mathcal{A}^c_{(\nu,s)}$ are defined as $\mathcal{A}_{(\nu,s)}=\{i
:i \in \{\!\{ N_s\}\!\},\:\: \text{and} \:\: N_s\nu+i \in \mathcal{I}\}$ and $\mathcal{A}^c_{(\nu,s)}=\{i
:i \in \{\!\{ N_s\}\!\},\:\: \text{and} \:\: N_s\nu+i \in \mathcal{I}^c\}$, respectively.
\end{definition}
\begin{example}
Based on Definition \ref{def_sets}, for the nodes $(0,3)$ and $(1,3)$ in Fig. \ref{bintree}, we have $\mathcal{A}_{(0,3)}=\{5,6,7\}$ and $\mathcal{A}^c_{(1,3)}=\{0\}$, respectively.
\end{example}
The following definition will prove helpful in Lemma \ref{lem1} and later.
\begin{definition}
The facilitator matrix $\boldsymbol{\mathsf{F}}_2^{\bigotimes k}=[\boldsymbol{\mathsf{f}}_0^{(k)}, \dots, \boldsymbol{\mathsf{f}}_{2^k-1}^{(k)}]$ with columns $\boldsymbol{\mathsf{f}}_i^{(k)}=[\mathsf{f}_{i,0}^{k},\dots,\mathsf{f}_{i,2^k-1}^{k}]^T, i=0,\dots,2^k-1$, is defined as the $k$th Kronecker power of matrix 
\begin{equation}
\boldsymbol{\mathsf{F}}_2 = \begin{bmatrix} 1 & 0 \\ 
\frac{\gamma}{\mu} &  1 \end{bmatrix}.
\end{equation}
\end{definition}

The following result is simple to prove.
\begin{lemma}
Consider $\mathbf{G}^{\bigotimes k}_2 = [\mathbf{g}^{(k)}_0, \dots, \mathbf{g}^{(k)}_{2^k-1}]$, $k \in \mathbb{Z^+}$, where $\mathbf{g}^{(k)}_{i}=[{g}^{(k)}_{i,0}\dots,{g}^{(k)}_{i,2^k-1}]^T$, $i=0,\dots, 2^k-1$, are column-vectors with $2^k$ elements that form the columns of matrix $\mathbf{G}^{\bigotimes k}_2$. Also, let the matrix $\overline{\mathbf{G}}^{\bigotimes k}_2 $ be defined as
\begin{equation}
\overline{\mathbf{G}}^{\bigotimes k}_2 = {\mathbf{G}}^{\bigotimes k}_2 \diag(\mathbf{g}^{(k)^{-1}}_0),\label{eq:Gbar}
\end{equation} 
where $\mathbf{g}^{(k)^{-1}}_0$ is the first column of $\mathbf{G}^{{\bigotimes k}^{-1}}_2$, i.e., the inverse of $\mathbf{G}^{\bigotimes k}_2$, and  $\diag(\mathbf{g}^{(k)^{-1}}_0)$ is a $2^k \times 2^k$ diagonal matrix whose diagonal entries are the $2^k$ elements of $\mathbf{g}^{(k)^{-1}}_0$. Then, all the nonzero elements of row $j$, $j = 0,\dots,2^k-1$, in $\overline{\mathbf{G}}^{\bigotimes k}_2$ are equal to $\mathsf{f}^{(k)}_{0,j}$. 
\label{lem1}
\end{lemma}
%\renewcommand\qedsymbol{$\blacksquare$}
%\begin{proof}
%by induction.
%\end{proof}
\begin{example}
Consider the generator matrix ${\mathbf{G}}^{\bigotimes 2}_2$. Then, we have
\begin{equation*}
\begin{split}
\overline{\mathbf{G}}^{\bigotimes 2}_2 &=
\overbrace{\left[ \begin{array}{c|c|c|r}
\mu^2& 0 & 0& 0\\ 
\mu\gamma & \mu\delta & 0 & 0\\
\mu\gamma & 0 & \mu\delta & 0 \\
\gamma^2 & \gamma\delta & \gamma\delta & \delta^2
\end{array}\right]}^{{\mathbf{G}}^{\bigotimes 2}_2}\times \overbrace{\left[ \begin{array}{cccr}
\frac{1}{\mu^2}& 0 & 0& 0\\ 
0 & \frac{\gamma}{\mu^2\delta} & 0 & 0\\
0 & 0 & \frac{\gamma}{\mu^2\delta} & 0 \\
0 & 0 & 0 & \frac{\gamma^2}{\mu^2\delta^2}
\end{array}\right]}^{\diag(\mathbf{g}^{(2)^{-1}}_0)}  \\
&= \left[ \begin{array}{cccr}
1& 0 & 0& 0\\ 
\frac{\gamma}{\mu} & \frac{\gamma}{\mu} & 0 & 0\\
\frac{\gamma}{\mu} & 0 & \frac{\gamma}{\mu} & 0 \\
\frac{\gamma^2}{\mu^2} & \frac{\gamma^2}{\mu^2} & \frac{\gamma^2}{\mu^2}& \frac{\gamma^2}{\mu^2}
\end{array}\right].
\end{split}
\end{equation*}
It can be seen that all the nonzero elements at rows $0$ to $3$ are equal to the corresponding elements in $\boldsymbol{\mathsf{f}}_0^{(2)}=[1,\frac{\gamma}{\mu},\frac{\gamma}{\mu},\frac{\gamma^2}{\mu^2}]^T$ at positions $0$ to $3$, respectively.
\end{example}
\begin{lemma}\label{lem_g0rN}
Suppose $\mathbf{g}^{(k)^{-1}}_{0}=[{g}^{(k)^{-1}}_{0,0}\dots,{g}^{(k)^{-1}}_{0,2^k-1}]^T$ and $\mathbf{r}^{(k)^T}_{2^k-1}=[{r}^{(k)}_{2^k-1,0},\dots,{r}^{(k)}_{2^k-1,2^k-1}]$ are the first column and last row of the matrices ${\mathbf{G}}^{{\bigotimes k}^{-1}}_2$ and ${\mathbf{G}}^{\bigotimes k}_2$, respectively. Then, 
\begin{equation}
\diag(\mathbf{g}^{(k)^{-1}}_0)\mathbf{r}^{(k)}_{2^k-1}=\frac{\gamma^k}{\mu^k}\mathbf{1}_{2^k\times 1},
\end{equation}
where $\gamma, \mu \in \mathbb{GF}(q)\backslash 0$, and $\mathbf{1}_{2^k\times 1}$ is an all-one vector with $2^k$ entries.
\end{lemma}
The following result can be easily proved using induction.
\begin{lemma}\label{lem_evennumber}
Suppose the vector $\mathbf{r}^{(k)^T}_{i}$ ($\overline{\mathbf{r}}^{(k)^T}_{i}$) denotes the $i$th row of the matrix ${\mathbf{G}}^{\bigotimes k}_2$ ($\overline{\mathbf{G}}^{\bigotimes k}_2$). The number of nonzero elements at each row $\mathbf{r}^{(k)^T}_{i}$ ($\overline{\mathbf{r}}^{(k)^T}_{i}$), $1 \leq i \leq 2^k-1$, is an even number ($i=0$ is excluded.).
\end{lemma}
In the following, the output vector of a node at the encoding side is denoted by $\mathbf{c}^{(\nu,s)}=[{c}_0^{(\nu,s)},\dots,{c}^{(\nu,s)}_{N_s-1}]^T$ while the estimated output-vector of a node at the decoding side is represented by $\boldsymbol{\chi}^{(\nu,s)}=[{\chi}_0^{(\nu,s)},\dots,{\chi}^{(\nu,s)}_{N_s-1}]^T$. Moreover, we use $\mathbf{u}_i^j=[u_i,\dots, u_j]^T$ to represent a slice of the input vector $\mathbf{u}$. Based on these notations, $\mathbf{c}^{(\nu,s)}$ is obtained as
\begin{equation}\label{eq_subcodwrd}
\mathbf{c}^{(\nu,s)^T} = \mathbf{u}_{N_s\nu}^{{N_s(\nu+1)-1}^T}{\mathbf{G}}^{\bigotimes s}_2.
\end{equation}
It is easy to verify that for $s=n$, \eqref{eq_subcodwrd} turns into Equation \eqref{eq:polarencfix}. In the rest of the paper, we may also refer to $\mathbf{c}^{(\nu,s)}$ and $\boldsymbol{\chi}^{(\nu,s)}$ as the output codeword and estimated codeword of node $(\nu,s)$, respectively.  
\begin{lemma}\label{lem:c_ctilde}
Let $\overline{\mathbf{G}}^{\bigotimes k}_2$ be defined as in Equation \eqref{eq:Gbar}. Also, let the vectors $\mathbf{c}^{(\nu,s)}$ and $\tilde{\mathbf{c}}^{(\nu,s)}$ be derived as
\begin{equation}
\begin{split}
\mathbf{c}^{(\nu,s)^T} = \mathbf{u}_{N_s\nu}^{{N_s(\nu+1)-1}^T}{\mathbf{G}}^{\bigotimes s}_2,\\
\tilde{\mathbf{c}}^{(\nu,s)^T} = \mathbf{u}_{N_s\nu}^{{N_s(\nu+1)-1}^T}\overline{\mathbf{G}}^{\bigotimes s}_2.
\end{split}
\end{equation}
Then, 
\begin{equation}
\begin{split}
&\tilde{\mathbf{c}}^{(\nu,s)}=\diag(\mathbf{g}^{(s)^{-1}}_0)\mathbf{c}^{(\nu,s)},\\
&{\boldsymbol{c}}^{(\nu,s)}=\gamma^{-s}\mu^s\diag(\mathbf{r}^{(s)}_{N_s-1})\tilde{\boldsymbol{c}}^{(\nu,s)}.
\end{split}
\end{equation}
\end{lemma}
Now that we are equipped with the above definitions and lemmas, in the following, we introduce different non-binary special nodes and their associated decoding algorithms. The proposed special nodes can significantly reduce the latency as well as the complexity of the LNBSC decoder. 
\subsection{Rate-0 Node}
A node $(\nu,s)$ is identified as a Rate-0 node if $\mathcal{A}_{(\nu,s)}=\{\}$. The output vector of a Rate-0 node is $\boldsymbol{c}^{(\nu,s)^T}=\mathbf{0}^T_{N_s\times 1}{\mathbf{G}}^{\bigotimes s}_2$, where $\mathbf{0}_{N_s\times 1}$ denotes an all-zero vector of size $N_s$. Hence, the estimated output vector is $\boldsymbol{\chi}^{(\nu,s)}=\mathbf{0}_{N_s\times 1}$. 

It is noted that if the left child of a node is a Rate-0 node, then the CN equation in \eqref{NBcnSPA} (or \eqref{NBcn}) can be skipped and the VN operation can be executed directly. 
\subsection{Rate-1 Node}
This node corresponds to $\mathcal{A}^c_{(\nu,s)}=\{\}$. Suppose the LLR matrix at the top of a Rate-1 node is $\mathbf{L}^{(\nu,s)}$. Then, the non-binary output vector is estimated as
\begin{equation}
\boldsymbol{\chi}^{(\nu,s)}=\mathcal{H}(\mathbf{L}^{(\nu,s)}),
\end{equation}
where $\mathcal{H}(\cdot)$ (defined in \eqref{NB_harddec}) represents the non-binary hard-decision on different columns of $\mathbf{L}^{(\nu,s)}$.
\subsection{M-REP Node}\label{sec_mrep}
%\begin{lemma}
A node $(\nu,s)$ is identified as a multiplicative repetition (M-REP) node if $\mathcal{A}_{(\nu,s)}=\{N_s-1\}$.
%\end{lemma}
%\begin{proof}
Suppose $\mathbf{r}^{(s)^T}_{N_s-1}=[{r}^{(s)}_{N_s-1,0},\dots,{r}^{(s)}_{N_s-1,N_s-1}]$ be the last row of the generator matrix ${\mathbf{G}}^{\bigotimes s}_2$. The output vector of an M-REP node is obtained as 
\begin{equation}
\mathbf{c}^{(\nu,s)}=u_{N_s(\nu+1)-1}\mathbf{r}^{(s)}_{N_s-1},
\label{mRepeq}
\end{equation}
where ${u}_{N_s(\nu+1)-1}$ is the symbol located at leaf-node position ${N_s(\nu+1)-1}$. It is observed that the output vector consists of symbol ${u}_{N_s(\nu+1)-1}$ multiplied by different elements of $\mathbf{r}^{(s)}_{N_s-1}$, hence the term ``multiplicative repetition''.
%\end{proof}
\begin{lemma}\label{lem_mrep2rep}
If a node $(\nu,s)$ is identified as an M-REP node, the vector $\tilde{\mathbf{c}}^{(\nu,s)}$ defined as 
\begin{equation}
\tilde{\mathbf{c}}^{(\nu,s)}=\diag(\mathbf{g}^{(s)^{-1}}_0)\mathbf{c}^{(\nu,s)},
\end{equation}
is a codeword of a repetition (REP) code with rate $1/N_s$.
\end{lemma}
\begin{proof}
Based on Lemma \ref{lem_g0rN}, multiplying Equation \eqref{mRepeq} from left by $\diag(\mathbf{g}^{(s)^{-1}}_0)$ results in
\begin{equation}\label{eq_cwrd_rep}
\tilde{\mathbf{c}}^{(\nu,s)}=\frac{\gamma^s}{\mu^s}{u}_{N_s(\nu+1)-1}\mathbf{1}_{N_s\times1}.
\end{equation}
It is observed that  $\tilde{\mathbf{c}}^{(\nu,s)}$ consists of $N_s$ repetitions of symbol $\frac{\gamma^s}{\mu^s}{u}_{N_s(\nu+1)-1}$ which makes it the codeword of a repetition code with rate $1/N_s$. 
\end{proof}
The result of Lemma \ref{lem_mrep2rep} can be used to decode an M-REP node. First, different columns of the LLR matrix $\mathbf{L}^{(\nu,s)}$ are permuted as
 \begin{equation}\label{eq_permutLLR}
 \tilde{\boldsymbol{\ell}}_i^{(\nu,s)}=\boldsymbol{\Pi}^\times_{g_{0,i}^{(s)^{-1}}}\boldsymbol{\ell}_i^{(\nu,s)},\: \:0\leq i\leq N_s-1,
 \end{equation}
where $\boldsymbol{\Pi}^\times_{g_{0,i}^{(s)^{-1}}}$ is a $q\times q$ permutation matrix corresponding to $i$th element of $\mathbf{g}^{(s)^{-1}}_0$. Then, a hard decision is made on the element-wise summation of the vectors $\tilde{\boldsymbol{\ell}}_i^{(\nu,s)}$ as
\begin{equation}\label{eq_hd_rep}
{\chi}'=\mathcal{H}(\sum_{i=0}^{N_s-1}\tilde{\boldsymbol{\ell}}_i^{(\nu,s)}),
\end{equation}
where $\chi'$ is an estimate of the symbol $\frac{\gamma^s}{\mu^s}{u}_{N_s(\nu+1)-1}$ in \eqref{eq_cwrd_rep}. The estimated vector $\boldsymbol{\chi}^{(\nu,s)}$ corresponding to the M-REP node output is then obtained as
\begin{equation}
\boldsymbol{\chi}^{(\nu,s)}=\chi' \gamma^{-s}\mu^s\mathbf{r}^{(s)}_{N_s-1}.\label{eq:mrepoutput}
\end{equation}
Note that $\gamma^{-s}$ is the inverse of $\gamma^s$ with respect to $\mathbb{GF}(q)$.
\subsection{M-SPC Node}\label{sec_mspc}
The multiplicative single parity check (M-SPC) node corresponds to $\mathcal{A}^c_{(\nu,s)}=\{0\}$.
\begin{theorem}\label{theo_mSPC}
Let $\mathbf{c}^{(\nu,s)}$ be the output of the node $(\nu,s)$ with  $\mathcal{A}^c_{(\nu,s)}=\{0\}$. Then, the elements of the vector $\tilde{\mathbf{c}}^{(\nu,s)}=\diag(\mathbf{g}^{(s)^{-1}}_0)\mathbf{c}^{(\nu,s)}$ form a non-binary SPC equation or, equivalently, the node $(\nu,s)$ is identified as a multiplicative SPC (M-SPC) node. 
\end{theorem}
\begin{proof}
 Multiplying \eqref{eq_subcodwrd} by $\diag(\mathbf{g}^{(s)^{-1}}_0)$ from right results in
\begin{equation}
\tilde{\mathbf{c}}^{(\nu,s)^T} = \mathbf{u}_{N_s\nu}^{{N_s(\nu+1)-1}^T}\overline{\mathbf{G}}^{\bigotimes s}_2,
\end{equation}
where $\overline{\mathbf{G}}^{\bigotimes s}_2$ is defined in Lemma \ref{lem1}. Since $\mathcal{A}^c_{(\nu,s)}=\{0\}$, the input vector $\mathbf{u}_{N_s\nu}^{{N_s(\nu+1)-1}} = [0, u_{N_s\nu+1},\dots,u_{N_s(\nu+1)-1}]^T$, meaning that the first row of $\overline{\mathbf{G}}^{\bigotimes s}_2$ would not contribute to $\tilde{\mathbf{c}}^{(\nu,s)}=[\tilde{{c}}^{(\nu,s)}_0,\dots,\tilde{{c}}^{(\nu,s)}_{N_s-1}]^T$. Also, according to Lemma \ref{lem1}, all the nonzero elements of row $j$, $0 \leq j \leq N_s-1$, in $\overline{\mathbf{G}}^{\bigotimes s}_2$ are equal. On the other hand, Lemma \ref{lem_evennumber} implies that the number of nonzero elements at each row of $\overline{\mathbf{G}}^{\bigotimes s}_2$ (excluding the first row) is even. Therefore, in the summation of elements in $\tilde{\mathbf{c}}^{(\nu,s)}$, different terms would cancel each other out and we have
\begin{equation}
\tilde{{c}}^{(\nu,s)}_0+\dots+\tilde{{c}}^{(\nu,s)}_{N_s-1}=0.
\end{equation}
\end{proof}
The result of Theorem \ref{theo_mSPC} can be used to decode an M-SPC node. In the first step, Equation \eqref{eq_permutLLR} is used to obtain the LLR matrix $\tilde{\mathbf{L}}^{(\nu,s)}$ from ${\mathbf{L}}^{(\nu,s)}$. Then, the proposed non-binary SPC decoder presented in Algorithm \ref{alg_nbSPC2} (it will be explained shortly) is used to compute the estimated vector $\tilde{\boldsymbol{\chi}}^{(\nu,s)}$ corresponding to $\tilde{\mathbf{c}}^{(\nu,s)}$. Finally, based on Lemma \ref{lem:c_ctilde}, the estimated output vector ${\boldsymbol{\chi}}^{(\nu,s)}$ is obtained as
\begin{equation}
{\boldsymbol{\chi}}^{(\nu,s)}=\gamma^{-s}\mu^s\diag(\mathbf{r}^{(s)}_{N_s-1})\tilde{\boldsymbol{\chi}}^{(\nu,s)},
\end{equation}
where $\diag(\mathbf{r}^{(s)}_{N_s-1})$ is a $2^s \times 2^s$ diagonal matrix whose diagonal entries are the $2^s$ elements in the last row of $\mathbf{G}^{\bigotimes s}_2$.
%\begin{remark}
The SPC nodes in the binary polar codes can be decoded using the Wagner decoding algorithm, \cite{hanif}, which is a maximum-likelihood (ML) decoder. A similar algorithm to Wagner can also be formulated for non-binary SPC nodes, but it is too complex. To avoid such complexity, we propose a low-complexity decoder for non-binary SPC nodes in Algorithm \ref{alg:nbSPC2} which has a remarkable error rate performance. This is while Algorithm \ref{alg:nbSPC2}, unlike Wagner decoder, is not an ML decoding algorithm.

According to Algorithm \ref{alg:nbSPC2}, first a hard decision is made on the input non-binary LLRs (line $3$). If the parity check equation is satisfied, the result of the hard decision, which is stored in vector $\boldsymbol{\chi}$, is returned as the output. In the case of nonzero parity, two different scenarios are considered. In the first scenario, it is assumed only one symbol can be erroneous and each of the elements in $\boldsymbol{\chi}$ are replaced by a symbol-candidate that satisfies the parity check equation (line $6$), and the corresponding symbol-LLR is stored (line $7$). Among the stored LLRs, the index of the minimum one, denoted by $\kappa$, is obtained (line $10$). %and the corresponding symbol in $\boldsymbol{\chi}$ is replaced by the symbol-candidate at position $\kappa$ (line $9$). The updated vector $\boldsymbol{\chi}$ is, finally, returned as the output of the decoder.

\begin{algorithm}[H]
\caption{Improved non-binary SPC decoding.}\label{alg:nbSPC2}
\begin{algorithmic}[1]
\STATE \colorbox{gray(x11gray)}{\textsc{Input:} matrix $\mathbf{L}=[\boldsymbol\ell_0,\dots,\boldsymbol\ell_{N_s-1}]$}
\STATE \colorbox{gray(x11gray)}{{\textsc{Output:}} vector $\boldsymbol{\chi}=[\chi_0,\dots,\chi_{N_s-1}]^T$}

\STATE \colorbox{gray(x11gray)}{$\boldsymbol{\chi} \gets \mathcal{H}(\mathbf{L})$}
\STATE \colorbox{gray(x11gray)}{\textbf{if} $\chi_0 + \dots + \chi_{N_s-1} \neq 0$}
%\STATE \hspace{0.5cm}$ \textbf{return }\boldsymbol{\chi}$
%\STATE \textbf{else}
%\STATE \hspace{0.5cm} ${\boldsymbol{\psi} = [\psi_0,\dots,\psi_{N_s-1}]^T} \gets \mathbf{0}_{N_s\times 1}$
\STATE \hspace{0.5cm} \colorbox{gray(x11gray)}{$\textbf{for } i=\{0,\dots,N_s-1\}$}
\STATE \hspace{1cm} \colorbox{gray(x11gray)}{${\psi}_i \gets\sum_{j=\{0,\dots,N_s-1\}\backslash i}\chi_j$}
\STATE \hspace{1cm} \colorbox{gray(x11gray)}{$w_i\gets \ell_{i,{\psi}_i}$}
\STATE \hspace{1cm} $\psi'_i\gets \underset{j\in \mathbb{GF}(q)\backslash \chi_i}{\argmin}({\ell}_{i,j})$
\STATE \hspace{1cm} $w'_i\gets {\ell}_{i,\psi'_i}$
%\STATE \hspace{1cm} $h_i\gets \underset{j\in \mathbb{GF}(q)\backslash \chi_i}{\min}({\ell}_{i,j})$
\STATE \hspace{0.5cm} \colorbox{gray(x11gray)}{$\kappa\gets\underset{m=\{0,\dots,N_s-1\}}{\argmin}(w_m)$}
\STATE \hspace{0.5cm} $\kappa'\gets\underset{m=\{0,\dots,N_s-1\}}{\argmin}(w'_m)$
\STATE \hspace{0.5cm} $\textbf{for } i=\{0,\dots,N_s-1\}\backslash \kappa'$
\STATE \hspace{1cm} $\bar{\psi}_i \gets\sum_{j=\{0,\dots,N_s-1\}\backslash i,\kappa'}\chi_j+\psi'_{\kappa'}$
\STATE \hspace{1cm} $\bar{w}_i\gets {\ell}_{i,\bar{\psi}_i}$
\STATE \hspace{0.5cm} $\bar{\kappa}\gets\underset{m=\{0,\dots,N_s-1\}}{\argmin}(\bar{w}_m)$
\STATE \hspace{0.5cm} \textbf{if} $\bar{w}_{\bar{\kappa}} + w'_{\kappa'}>w_\kappa$
\STATE \hspace{1cm} \colorbox{gray(x11gray)}{$\chi_\kappa\gets\psi_\kappa$}
\STATE \hspace{0.5cm} \textbf{else}
\STATE \hspace{1cm} $\chi_{\kappa'}\gets\psi'_{\kappa'}$, $\chi_{\bar{\kappa}}\gets\bar{\psi}_{\bar{\kappa}}$
\STATE $ \textbf{return }\boldsymbol{\chi}$
\end{algorithmic}
\label{alg_nbSPC2}
\end{algorithm}
%While Algorithm \ref{alg:nbSPC} has low complexity, our simulation results show that using it for the decoding of M-SPC nodes in the fast-LNBSC decoder might result in $0.1-0.3$dB loss in error rate compared to LNBSC decoder without special nodes. To reduce this loss, we propose Algorithm \ref{alg:nbSPC2} which, according to the simulation results, has a remarkable error rate performance and its loss is almost $0$.  
In the second scenario, in addition to the single error cases in the first scenario, a special case of $2$ symbol errors is also considered. In lines $8-9$ of Algorithm \ref{alg:nbSPC2}, the second minimum\footnote{It is noted that the first minimum of an LLR vector corresponds to the hard decision.} of each LLR vector along with its corresponding symbol is stored. Then, at line $11$, the index of the symbol that has the smallest 2nd-minimum LLR is found and stored in variable $\kappa'$. At this stage, it is assumed that the symbol at position $\kappa'$ of $\boldsymbol{\chi}$ is equal to the symbol associated with the second minimum. Then, each of the elements in $\boldsymbol{\chi}$, except the one at position $\kappa'$, are one-by-one replaced by a symbol-candidate that satisfies the parity check equation (line $13$), and the corresponding symbol-LLR is stored (line $14$). Among the stored LLRs, the index of minimum one, denoted by $\bar{\kappa}$, is obtained (line $15$). In the last step, if the minimum LLR obtained for single-error case (line $10$) is less than the summation of minimum LLR obtained at line $15$ and the 2nd-minimum LLR corresponding to position $\kappa'$, then the symbol at position $\kappa$ in $\boldsymbol{\chi}$ is replaced by the corresponding symbol-candidate from line $6$. Otherwise, two symbols of $\boldsymbol{\chi}$ at positions $\kappa'$ and $\bar{\kappa}$ are, simultaneously, replaced by the corresponding symbols found in lines $8$ and $13$. We note that one could only track the shaded gray lines (single error scenario) in Algorithm \ref{alg:nbSPC2} to further reduce complexity, but our investigations showed that this is not good enough. In the following, we present an example to clarify the steps taken in Algorithms \ref{alg:nbSPC2}.

\begin{example}
Suppose $f(x)=x^2+x+1$ is adopted as the primitive polynomial of $\mathbb{GF}(4)$. Consider the input to the non-binary SPC decoders in Algorithm \ref{alg:nbSPC2} is the matrix
\begin{equation*}
\boldsymbol{L}=\left[ \begin{array}{cccr}
0& 5 & 17& 0\\ 
12 & 10 & 0 & 8\\
34& 0 & 14 & 25 \\
6 & 63 & 16 & 33
\end{array}\right].
\end{equation*}
Using the decimal representation of the non-binary field elements, the hard-decision output at line $2$ of the algorithm is $\boldsymbol{\chi} = [0,2,1,0]^T$. Since the parity equation is not satisfied ($0+2+1+0\neq 0$), Algorithm \ref{alg:nbSPC2} in the first scenario considers $4$ candidate symbols, corresponding to $4$ positions, which satisfy the parity equation, i.e., symbols $\psi_0 = 3, \psi_1=1,\psi_2=2$ and $\psi_3=3$ ( $\boldsymbol{\psi}=[3,1,2,3]^T$), and stores their corresponding LLRs, i.e., $\boldsymbol{\omega}=[6,10,14,33]^T$. Because $w_0$ is the minimum LLR in $\boldsymbol{\omega}$, $\psi_0=3$ is selected to substitute $\chi_0$.% that results in $\boldsymbol{\chi}=[3,2,1,0]^T$.

Based on Algorithm \ref{alg:nbSPC2}, in the second scenario, symbols corresponding to second minimum-LLRs are also considered, i.e., $\boldsymbol{\psi}'=[3,0,2,1]^T$ with LLRs $\boldsymbol{\omega}'=[6,5,14,8]^T$. Since $w'_1=5$ has the minimum value among the elements of $\boldsymbol{\omega}'$, it is assumed that $\chi_1=\psi'_1=0$. Now, with the exception of position $1$, three candidate symbols for the three remaining positions are obtained, i.e., $\bar{\boldsymbol{\psi}}=[1,?,0,1]^T$ with LLRs $\bar{\boldsymbol{\omega}}=[12,?,17,8]^T$. $\bar{w}_3=8$ is the minimum LLR in $\bar{\boldsymbol{\omega}}$, so, $\bar{\psi}_3=1$ is considered as a candidate to replace $\chi_3$. At this stage, we have two options for output vector $\boldsymbol{\chi}$. In the first one, only $\chi_0$ is replaced by $\psi_0=3$ with LLR $w_0=6$. In the second option, $\chi_1$ and $\chi_3$ are replaced by $\psi_1'=0$ (with LLR $w'_1=5$) and $\bar{\psi}_3=1$ (with $\bar{w}_3=8$), respectively. Since $(w'_1+\bar{w}_3=13) > (w_0=8)$, due to the lower penalty, the first option is chosen (line 17), i.e., $\boldsymbol{\chi} = [\bold{3},2,1,0]^T.$ 
\end{example}
\subsection{Type-I Node}
A node $(\nu,s)$ is identified as a Type-I node if $\mathcal{A}_{(\nu,s)}=\{N_s-2,N_s-1\}$.

The following result is simple to prove.
\begin{theorem}
Let $\mathbf{c}^{(\nu,s)}_e$ and $\mathbf{c}^{(\nu,s)}_o$ be the symbols located at even and odd indexes of the output codeword of a Type-I node. Then, the elements of $\mathbf{c}^{(\nu,s)}_e$ and $\mathbf{c}^{(\nu,s)}_o$ form two separate M-REP code as
\begin{equation}
\begin{split}
\mathbf{c}^{(\nu,s)}_e &= (\mu u_{N_s(\nu+1)-2}+\gamma u_{N_s(\nu+1)-1})\mathbf{r}^{(s-1)}_{N_s/2-1},\\
\mathbf{c}^{(\nu,s)}_o &= \delta u_{N_s(\nu+1)-1}\mathbf{r}^{(s-1)}_{N_s/2-1},
\end{split}
\end{equation}
where $\mathbf{r}^{(s-1)}_{N_s/2-1}$ is the last row of the generator matrix $\mathbf{G}^{\bigotimes s-1}_2$.
\end{theorem}
In order to decode a Type-I node, the matrix $\mathbf{L}^{(\nu,s)}$ is divided into two submatrices,  $\mathbf{L}^{(\nu,s)}_e$ and $\mathbf{L}^{(\nu,s)}_o$, that consist of the even and odd indexed columns of $\mathbf{L}^{(\nu,s)}$, respectively. Then, $\mathbf{L}^{(\nu,s)}_e$ and $\mathbf{L}^{(\nu,s)}_o$ along with the M-REP decoding of Section \ref{sec_mrep} are used to decode two separate M-REP codes with length $N_s/2$. Eventually, the estimated output vectors, denoted by $\boldsymbol{\chi}^{(\nu,s)}_e$ and $\boldsymbol{\chi}^{(\nu,s)}_o$, are placed at the corresponding even and odd positions of ${\boldsymbol{\chi}}^{(\nu,s)}$ to obtain the total output vector.
\subsection{Type-II Node}
A node $(\nu,s)$ is identified as a Type-II node if $\mathcal{A}_{(\nu,s)}=\{N_s-3,N_s-2,N_s-1\}$.
\begin{theorem}\label{theo_T2}
If a node $(\nu,s)$ is identified as a Type-II node, its output vector is
\begin{equation}
\mathbf{c}^{(\nu,s)^T}=[{r}^{(s-2)}_{N_s/4-1,0}\mathbf{c}^{(2)^T},\dots,{r}^{(s-2)}_{N_s/4-1,N_s/4-1}\mathbf{c}^{(2)^T}],
\label{eq_T2}
\end{equation} 
where $\mathbf{r}^{(s-2)}_{N_s/4-1}=[{r}^{(s-2)}_{N_s/4-1,0},\dots,{r}^{(s-2)}_{N_s/4-1,N_s/4-1}]^T$ is the last row of $\mathbf{G}^{\bigotimes s-2}_2$ and $\mathbf{c}^{(2)}=[{c}^{(2)}_0,{c}^{(2)}_1,{c}^{(2)}_2,{c}^{(2)}_3]^T$ is the output vector of an M-SPC node with length $4$, i.e., $g^{(2)^{-1}}_{0,0}{c}^{(2)}_0+g^{(2)^{-1}}_{0,1}{c}^{(2)}_1+g^{(2)^{-1}}_{0,2}{c}^{(2)}_2+g^{(2)^{-1}}_{0,3}{c}^{(2)}_3=0$.
\end{theorem}
\begin{proof}
The theorem can be proved using mathematical induction. For $s=2$, $\mathbf{c}^{(\nu,2)}=\mathbf{c}^{(2)}$ which is an M-SPC node. For $s\geq 3$, we have $\mathbf{c}^{(\nu,s)^T}=[\mathbf{0}_{2^{s-1}\times 1}^T
+\gamma \mathbf{c}^{(\nu,s-1)^T},\delta\mathbf{c}^{(\nu,s-1)^T}]$ whose non-recursive form can be written as \eqref{eq_T2}.
\end{proof}
In order to decode a Type-II node, the columns of the soft input matrix $\mathbf{L}^{(\nu,s)}$ are permuted as
\begin{equation}
\dot{\boldsymbol{\ell}}_i^{(\nu,s)}=\boldsymbol{\Pi}^{\times^{-1}}_{r_{N_s/4-1,\lfloor i/4 \rfloor}^{(s-2)}}\boldsymbol{\ell}_i^{(\nu,s)},\: \:0\leq i\leq N_s-1.
\end{equation}
Then, the columns of LLR matrix $\mathbf{L}^{(2)}=[\boldsymbol{\ell}^{(2)}_0,\dots,\boldsymbol{\ell}^{(2)}_3]$ are obtained as
\begin{equation}
\boldsymbol{\ell}^{(2)}_j=\sum_{k=0}^{N_s/4-1}\dot{\boldsymbol{\ell}}_{4k+j}^{(\nu,s)}, \: j=0,\dots,3.\label{eq:T2summationllr}
\end{equation} 
The calculated matrix $\mathbf{L}^{(2)}$ is then used as the input of an M-SPC decoder (explained in Section \ref{sec_mspc}) to obtain the estimated vector $\boldsymbol{\chi}^{(2)}=[{\chi}^{(2)}_0,\dots,{\chi}^{(2)}_3]^T$. The output vector, $\boldsymbol{\chi}^{(\nu,s)}$, is then computed as
\begin{equation}
\boldsymbol{\chi}^{(\nu,s)^T}=[{r}^{(s-2)}_{N_s/4-1,0}\boldsymbol{\chi}^{(2)^T},\dots,{r}^{(s-2)}_{N_s/4-1,N_s/4-1}\boldsymbol{\chi}^{(2)^T}].
\end{equation}  
\subsection{Type-III Node}
A node $(\nu,s)$, $s\geq 2$, is identified as a Type-III node if $\mathcal{A}_{(\nu,s)}^c=\{0,1\}$.
\begin{theorem}
Let $\mathbf{c}^{(\nu,s)}_e$ and $\mathbf{c}^{(\nu,s)}_o$ be the symbols located at even and odd indexes of the output vector of a Type-III node. Then, the elements of the vectors $\tilde{\mathbf{c}}^{(\nu,s)}_e=\diag(\mathbf{g}^{(s-1)^{-1}}_0)\mathbf{c}^{(\nu,s)}_e$ and $\tilde{\mathbf{c}}^{(\nu,s)}_o=\diag(\mathbf{g}^{(s-1)^{-1}}_0)
\mathbf{c}^{(\nu,s)}_o$ form two separate non-binary SPC equations with length $N_s/2$, i.e., 
\begin{equation}\label{eq_T3}
\sum_{i=0}^{N_s/2-1}{\tilde{{c}}^{(\nu,s)}_{e,i}}=
\sum_{i=0}^{N_s/2-1}{\tilde{{c}}^{(\nu,s)}_{o,i}}=0
\end{equation}
\end{theorem}
\begin{proof}
The theorem can be proved using mathematical induction. For $s=2$, the output codeword can be written as $\mathbf{c}^{(\nu,2)}=[\gamma u_2, \gamma u_3, \delta u_2, \delta u_3]^T$ that results in $\tilde{\mathbf{c}}^{(\nu,2)}_{e}=[\frac{\gamma}{\mu} u_2, \frac{\gamma}{\mu} u_2]^T$ and $\tilde{\mathbf{c}}^{(\nu,2)}_{o}=[ \frac{\gamma}{\mu} u_3, \frac{\gamma}{\mu} u_3]^T$ based on which \eqref{eq_T3} is satisfied. For $s>2$, we presume that \eqref{eq_T3} is satisfied for polar codes with length $N_s/2$. Suppose Equation \eqref{eq_T3} is satisfied for the left child node of $(\nu,s)$, i.e., $\mathbf{c}^{(2\nu,s-1)}$. For node $(\nu,s)$, we have $\mathbf{c}^{(\nu,s)^T} = [\mu \mathbf{c}^{(2\nu,s-1)^T} + \gamma\mathbf{c}^{(2\nu+1,s-1)^T},\delta \mathbf{c}^{(2\nu+1,s-1)^T}]$. Based on the fact that $\mathbf{g}_0^{(s-1)^{-1}}=[\frac{1}{\mu} \mathbf{g}_0^{(s-2)^{-1^T}},\frac{\gamma}{\mu \delta} \mathbf{g}_0^{(s-2)^{-1^T}}]^T$, we can write $\tilde{\mathbf{c}}^{(\nu,s)^T}_{e/o} = [\tilde{\mathbf{c}}^{(2\nu,s-1)^T}_{e/o} + \frac{\gamma}{\mu}\tilde{\mathbf{c}}^{(2\nu+1,s-1)^T}_{e/o},\frac{\gamma}{\mu}\tilde{\mathbf{c}}^{(2\nu+1,s-1)^T}_{e/o}]$. Hence, for even/odd indexes, we have
\begin{equation}
\begin{split}
\sum_{j=0}^{N_s/2-1}{\tilde{{c}}^{(\nu,s)}_{e/o,j}}&=\sum_{k=0}^{N_s/4-1}\big({\tilde{{c}}^{(2\nu,s-1)}_{e/o,k}+\frac{\gamma}{\mu} \tilde{{c}}^{(2\nu+1,s-1)}_{e/o,k}}\\&+\frac{\gamma}{\mu} \tilde{{c}}^{(2\nu+1,s-1)}_{e/o,k}\big )=\sum_{k=0}^{N_s/4-1}{\tilde{{c}}^{(2\nu,s-1)}_{e/o,k}}=0.
\end{split}
\end{equation}
\end{proof}
In order to decode a Type-III node, the matrix $\mathbf{L}^{(\nu,s)}$ is divided into two submatrices,  $\mathbf{L}^{(\nu,s)}_e$ and $\mathbf{L}^{(\nu,s)}_o$, that consist of the even and odd indexed columns of $\mathbf{L}^{(\nu,s)}$, respectively. Then, $\mathbf{L}^{(\nu,s)}_e$ and $\mathbf{L}^{(\nu,s)}_o$ are used as the input of M-SPC decoder of Section \ref{sec_mspc} to decode two separate length-$N_s/2$ M-SPC codes in parallel. The resulted estimations, denoted by $\boldsymbol{\chi}^{(\nu,s)}_e$ and $\boldsymbol{\chi}^{(\nu,s)}_o$, are then placed at even and odd positions of ${\boldsymbol{\chi}}^{(\nu,s)}$, respectively, to obtain the total estimated output-vector.
\subsection{Type-IV Node}
A node $(\nu,s)$, $s\geq 2$, is identified as a Type-IV node if $\mathcal{A}_{(\nu,s)}^c=\{0,1,2\}$.
\begin{theorem}
Let a node $(\nu,s)$ be identified as a Type-IV node. Also, suppose the vector $\overline{\mathbf{g}}$ be defined as $\overline{\mathbf{g}}=[{g}_{0,0}^{(s-2)^{-1}}\mathbf{1}_{N_s/4\times 1}^T,\dots,{g}_{0,N_s/4-1}^{(s-2)^{-1}}\mathbf{1}_{N_s/4\times 1}^T]^T$. Then,
\begin{equation}
\begin{split}
\sum_{i=0}^{N_s/4-1}{\grave{{c}}^{(\nu,s)}_{4i+k}}&=r^{(2)}_{3,k} \rho,\:\:\: k=0,\dots,3,
\end{split}
\label{eq_T4}
\end{equation}
where ${\grave{\mathbf{c}}^{(\nu,s)}}$ is defined as $\grave{\mathbf{c}}^{(\nu,s)}=\diag(\overline{\mathbf{g}})\mathbf{c}^{(\nu,s)}$. Also, $\rho \in \mathbb{GF}(q)$ and $r_{3,k}^{(2)}$ is the $k$th element of $\mathbf{r}^{(2)}_3$, i.e., last row of $\mathbf{G}^{\bigotimes 2}_2$. 
\end{theorem}
\begin{proof}
We employ mathematical induction to prove this theorem. For $s=2$, \eqref{eq_T4} holds because node $(\nu,s)$ turns into an M-REP node. For $s>2$, we presume that \eqref{eq_T4} holds for block size $N_s/2$.  Let the vector $\mathbf{c}^{(\nu,s)^T}$ be partitioned into $4$ subvectors as $\mathbf{c}^{(\nu,s)^T}_k=[{c}^{(\nu,s)}_{k}, {c}^{(\nu,s)}_{4+k},\dots,{c}^{(\nu,s)}_{N_s-4+k}]$, $k=0,\dots,3$. Suppose Equation \eqref{eq_T4} is satisfied for the left child node of $(\nu,s)$, i.e., $\mathbf{c}^{(2\nu,s-1)}$. For node $(\nu,s)$, we have $\mathbf{c}^{(\nu,s)^T} = [\mu \mathbf{c}^{(2\nu,s-1)^T} + \gamma\mathbf{c}^{(2\nu+1,s-1)^T},\delta \mathbf{c}^{(2\nu+1,s-1)^T}]$ or, equivalently, $\mathbf{c}^{(\nu,s)^T}_k =[\mu \mathbf{c}^{(2\nu,s-1)^T}_k + \gamma\mathbf{c}^{(2\nu+1,s-1)^T}_k,\delta \mathbf{c}^{(2\nu+1,s-1)^T}_k]$, $k=0,\dots,3$. Based on the fact that $\mathbf{g}_0^{(s-2)^{-1}}=[\frac{1}{\mu}\mathbf{g}_0^{(s-3)^{-1^T}},\frac{\gamma}{\mu \delta} \mathbf{g}_0^{(s-3)^{-1^T}}]^T$, we can write $\grave{\mathbf{c}}^{(\nu,s)^T}_k =[\grave{\mathbf{c}}^{(2\nu,s-1)^T}_k + \frac{\gamma}{\mu}\grave{\mathbf{c}}^{(2\nu+1,s-1)^T}_k,\frac{\gamma}{\mu}\grave{\mathbf{c}}^{(2\nu+1,s-1)^T}_k]$. It is then easy to verify that 
\begin{equation}
\sum_{i=0}^{N_s/4-1}\grave{{c}}^{(\nu,s)}_{4i+k}=\sum_{r=0}^{N_s/8-1}\grave{{c}}^{(2\nu,s-1)}_{4r+k}=\rho r^{(2)}_{3,k},\: k=0,\dots,3,
\end{equation}
which completes the proof.
\end{proof}
A Type-IV node has a length-$4$ M-REP node as its left most descendant whose output is the vector $\rho \mathbf{r}^{(2)}_3$. Hence, to decode a Type-IV node, the LNBSC decoding continues until the M-REP node is decoded. Then, the four partitions of LLR matrix defined as $\mathbf{L}^{(\nu,s)}_k=[\boldsymbol{\ell}^{(\nu,s)}_{k}, \boldsymbol{\ell}^{(\nu,s)}_{4+k},\dots,\boldsymbol{\ell}^{(\nu,s)}_{N_s-4+k}]$, $k=0,\dots,3$, along with the M-REP node output are used to decoded $4$ separate M-SPC equations with size $N_s/4$ in parallel. Finally, the estimated vectors, according to their partitions are concatenated to obtain the total output of a Type-IV node.
\subsection{Type-V Node}
A node $(\nu,s)$, $s\geq 2$, is identified as a Type-V node if $\mathcal{A}_{(\nu,s)}=\{N_s-5, N_s-3, N_s-2, N_s-1\}$.
\begin{theorem}
If a node $(\nu,s)$ is identified as a Type-V node, its output vector is
\begin{equation}
\mathbf{c}^{(\nu,s)^T}=[{r}^{(s-3)}_{N_s/8-1,0}\mathbf{c}^{(r)^T},\dots,{r}^{(s-3)}_{N_s/8-1,N_s/8-1}\mathbf{c}^{(r)^T}],
\label{eq_T5}
\end{equation} 
where $\mathbf{r}^{(s-3)}_{N_s/8-1}=[{r}^{(s-3)}_{N_s/8-1,0},\dots,{r}^{(s-3)}_{N_s/8-1, N_s/8-1}]^T$ is the last row of $\mathbf{G}^{\bigotimes s-3}_2$ and $\mathbf{c}^{(r)}=[{c}^{(r)}_0,\dots,{c}^{(r)}_7]^T$ is the output vector of a node with size $8$ whose left and right children are M-REP and M-SPC nodes, respectively.
\end{theorem}
\begin{proof}
The theorem can be proved using a similar approach as in the proof of Theorem \ref{theo_T2}.
\end{proof}
In order to decode a Type-V node, the columns of the soft input matrix $\mathbf{L}^{(\nu,s)}$ are permuted as
\begin{equation}
\ddot{\boldsymbol{\ell}}_i^{(\nu,s)}=\boldsymbol{\Pi}^{\times^{-1}}_{r_{N_s/8-1,\lfloor i/8 \rfloor}^{(s-3)}}\boldsymbol{\ell}_i^{(\nu,s)},\: \:0\leq i\leq N_s-1.
\end{equation}
Then, the columns of LLR matrix $\mathbf{L}^{(r)}=[\boldsymbol{\ell}^{(r)}_0,\dots,\boldsymbol{\ell}^{(r)}_7]$ are obtained as
\begin{equation}
\boldsymbol{\ell}^{(r)}_j=\sum_{k=0}^{N_S/8-1}\ddot{\boldsymbol{\ell}}_{8k+j}^{(\nu,s)}, \: j=0,\dots,7.\label{eq:sumllr_T5}
\end{equation} 
The calculated matrix $\mathbf{L}^{(r)}$ is then used as the input of the concatenated (M-REP)-(M-SPC) node to obtain the estimated vector $\boldsymbol{\chi}^{(r)}=[{\chi}^{(r)}_0,\dots,{\chi}^{(r)}_7]^T$. The output vector, $\boldsymbol{\chi}^{(\nu,s)}$, is eventually computed as
\begin{equation}
\boldsymbol{\chi}^{(\nu,s)^T}=[{r}^{(s-3)}_{N_s/8-1,0}\boldsymbol{\chi}^{(r)^T},\dots,{r}^{(s-3)}_{N_s/8-1,N_s/8-1}\boldsymbol{\chi}^{(r)^T}].
\end{equation}  
\subsection{GM-REP Node}
\begin{figure}[!t]
\centering
\includegraphics[width=3.5in]{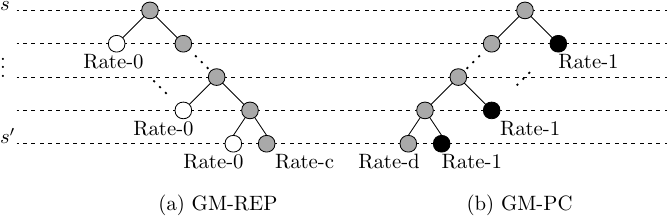}
\caption{General structures of (a) GM-REP  and (b) GM-PC  nodes.}
\label{fig_gmpc_gmrep}
\end{figure}
The structure of a generalized multiplicative repetition (GM-REP) node is shown in Fig.  \ref{fig_gmpc_gmrep}.a. A node $(\nu,s)$ at level $s$ is identified as a GM-REP node for which all its descendants are Rate-0 nodes except the rightmost descendant at level $s'<s$ that is a generic node with rate $c$ (Rate-c). We refer to the Rate-c node as the \emph{source node} of a GM-REP node.
\begin{theorem}
Let $\mathbf{c}^{(s')}$ with size $N_{s'}=2^{s'}$ denotes the output vector of the source node in a GM-REP node. Then, the vector $\mathbf{c}^{(\nu,s)}$ at the top of the GM-REP node can be obtained as
\begin{equation}
\mathbf{c}^{(\nu,s)^T}=[{r}^{(s-s')}_{2^{s-s'}-1,0}\mathbf{c}^{(s')^T},\dots,{r}^{(s-s')}_{2^{s-s'}-1,2^{s-s'}-1}\mathbf{c}^{(s')^T}],
\label{eq_grep}
\end{equation}
where $\mathbf{r}^{(s-s')}_{2^{s-s'}-1}=[{r}^{(s-s')}_{2^{s-s'}-1,0},\dots,{r}^{(s-s')}_{N_s/N_{s'}-1,2^{s-s'}-1}]^T$ is the last row of $\mathbf{G}^{\bigotimes s-s'}_2$.
\end{theorem}
\begin{proof}
The proof is similar to Theorem \ref{theo_T2}.
\end{proof}
In order to decode a GM-REP node, the columns of the soft input matrix $\mathbf{L}^{(\nu,s)}$ are permuted as
\begin{equation}
\check{\boldsymbol{\ell}}_i^{(\nu,s)}=\boldsymbol{\Pi}^{\times^{-1}}_{r_{2^{s-s'}-1,\lfloor i/N_{s'} \rfloor}^{(s-s')}}\boldsymbol{\ell}_i^{(\nu,s)},\: \:0\leq i\leq N_s-1.
\end{equation}
Then, the columns of LLR matrix $\mathbf{L}^{(s')}=[\boldsymbol{\ell}^{(s')}_0,\dots,\boldsymbol{\ell}^{(s')}_{N_{s'}-1}]$ are obtained as
\begin{equation}
\boldsymbol{\ell}^{(s')}_j=\sum_{k=0}^{2^{s-s'}-1}\check{\boldsymbol{\ell}}_{N_{s'}k+j}^{(\nu,s)}, \: j=0,\dots,N_{s'}-1.
\end{equation}
Matrix $\mathbf{L}^{(s')}$ is then used as the input of the source node to obtain the estimated vector $\boldsymbol{\chi}^{(s')}=[{\chi}^{(s')}_0,\dots,{\chi}^{(s')}_{N_{s'}}]^T$. The output vector, $\boldsymbol{\chi}^{(\nu,s)}$, is finally computed as
\begin{equation}
\boldsymbol{\chi}^{(\nu,s)^T}=[{r}^{(s-s')}_{2^{s-s'}-1,0}\boldsymbol{\chi}^{(s')^T},\dots,{r}^{(s-s')}_{2^{s-s'}-1,2^{s-s'}-1}\boldsymbol{\chi}^{(s')^T}].
\end{equation}  
\begin{remark}
Type-I, Type-II and Type-V nodes are special instances of the GM-REP node. 
\end{remark}
\begin{remark}
\textcolor{black}{The source-node in a GM-REP node can be a generic node with any rate. However, in this paper, a node is identified as a GM-REP node only if its source-node be a special node.} 
\end{remark}
\begin{remark}
The G-REP nodes in BPCs are generalized in \cite{SR} where the SR nodes are proposed. The left descendants of SR nodes include both Rate-0 and REP nodes. We could, similarly, define non-binary SR nodes for NBPCs. However, the high complexity of non-binary SR nodes prevented us from doing so. In fact, if the number of left-descendants that are REP nodes be $n_r$, then the number of parallel source-node decoders in BPCs is $2^{n_r}$. This is while the number of parallel source-node decoding for NBPCs is $q^{n_r}$ which can be a huge number. 
\end{remark}
\subsection{GM-PC Node}
The general structure of a generalized multiplicative parity check (GM-PC) node is shown in Fig.  \ref{fig_gmpc_gmrep}.b. A node $(\nu,s)$ at level $s$ is identified as a GM-PC node if all its descendants are Rate-1 nodes except the leftmost descendant at level $s'<s$ that is a generic node with rate $d$ (Rate-d). We refer to the Rate-d node as the \emph{parity node} of a GM-PC node.
\begin{theorem}
Let $\mathbf{c}^{(s')}$ with size $N_{s'}=2^{s'}$ denotes the output vector of the parity node in a GM-PC node. Also, suppose the vector $\overline{\overline{\mathbf{g}}}$ be defined as $\overline{\overline{\mathbf{g}}}=[{g}_{0,0}^{(s-s')^{-1}}\mathbf{1}_{2^{s-s'}\times 1}^T,\dots,{g}_{0,2^{s-s'}-1}^{(s-s')^{-1}}\mathbf{1}_{2^{s-s'}\times 1}^T]^T$. Then,
\begin{equation}
\sum_{i=0}^{2^{s-s'}-1}{\grave{\grave{{c}}}^{(\nu,s)}_{iN_{s'}+j}}={c}^{(s')}_j,\:\:j=0,\dots,N_{s'}-1,
\label{eq_gmpc}
\end{equation}
where $\grave{\grave{{c}}}^{(\nu,s)}_k$ are the elements of the vector $\grave{\grave{\mathbf{c}}}^{(\nu,s)}=\diag(\overline{\overline{\mathbf{g}}})\mathbf{c}^{(\nu,s)}$.
\end{theorem}
\begin{proof}
We employ mathematical induction to prove this theorem. For $s=s'$, \eqref{eq_gmpc} holds because node $(\nu,s)$ is equivalent to its parity node with codeword $\mathbf{c}^{(s')}$. For $s>s'$, we presume that \eqref{eq_gmpc} holds for block size $N_s/2$.  Let the vector $\mathbf{c}^{(\nu,s)^T}$ be partitioned into $N_{s'}$ subvectors as $\mathbf{c}^{(\nu,s)^T}_k=[{c}^{(\nu,s)}_{k}, {c}^{(\nu,s)}_{N_{s'}+k}, {c}^{(\nu,s)}_{2N_{s'}+k},\dots,{c}^{(\nu,s)}_{N_s-N_{s'}+k}]$, $k=0,\dots,N_{s'}$. Suppose Equation \eqref{eq_gmpc} is satisfied for the left child node of $(\nu,s)$, i.e., $\mathbf{c}^{(2\nu,s-1)}$. For node $(\nu,s)$, we have $\mathbf{c}^{(\nu,s)^T} = [\mu \mathbf{c}^{(2\nu,s-1)^T} + \gamma\mathbf{c}^{(2\nu+1,s-1)^T},\delta \mathbf{c}^{(2\nu+1,s-1)^T}]$ or, equivalently, $\mathbf{c}^{(\nu,s)^T}_k =[\mu \mathbf{c}^{(2\nu,s-1)^T}_k + \gamma\mathbf{c}^{(2\nu+1,s-1)^T}_k,\delta \mathbf{c}^{(2\nu+1,s-1)^T}_k]$, $k=0,\dots,N_{s'}-1$. Since $\mathbf{g}_0^{(s-s')^T}=[\frac{1}{\mu}\mathbf{g}_0^{(s-s'-1)^T},\frac{\gamma}{\mu \delta} \mathbf{g}_0^{(s-s'-1)^T}]$, we can write $\grave{\grave{\mathbf{c}}}^{(\nu,s)^T}_k =[\grave{\grave{\mathbf{c}}}^{(2\nu,s-1)^T}_k + \frac{\gamma}{\mu}\grave{\grave{\mathbf{c}}}^{(2\nu+1,s-1)^T}_k,\frac{\gamma}{\mu}\grave{\grave{\mathbf{c}}}^{(2\nu+1,s-1)^T}_k]$. It is then easy to verify that 
\begin{equation}
\sum_{i=0}^{2^{s-s'}-1}\grave{\grave{{c}}}^{(\nu,s)}_{iN_{s'}+j}=\sum_{r=0}^{2^{s-s'-1}-1}\grave{\grave{{c}}}^{(2\nu,s-1)}_{rN_{s'}+j},\: j=0,\dots,N_{s'}-1,
\end{equation}
which completes the proof.
\end{proof}
\begin{remark}
Type-III and Type-IV nodes are special instances of the GM-PC node. 
\end{remark}
\begin{remark}
\textcolor{black}{The parity-node in a GM-PC node can be a generic node with any rate. However, in this paper, a node is identified as a GM-PC node only if its parity-node be a special node. }
\end{remark}
\section{Proposed Simplified NBPC structure\\ Suitable for Fast Decoding}\label{sec:Gnbpc}
\begin{figure}[!t]
\centering
\includegraphics[width=3in]{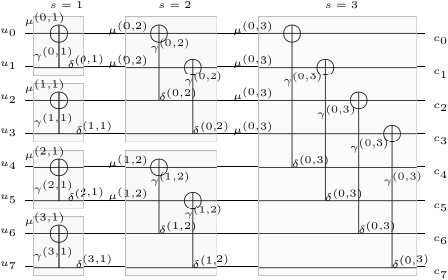}
\caption{Example of the general structure of NBPCs with variable kernel-coefficients in polarization units. Each gray box corresponds to a node at level $s$ of the binary-tree of the polar code.}\label{fig:simplifiedNBPC}
\end{figure}
\subsection{Proposed NBPC Structure}
%\begin{figure*}
%\centering
%\includegraphics[width=5in]{systencGNBP.eps}
%\caption{Proposed systematic encoding.}
%\label{fig:systenc}
%\end{figure*}

In the previous sections, according to Equation \eqref{eq:polarencfix}, it is assumed that the kernel coefficients in all the polarization stages are equal to fixed values, denoted by $\mu$, $\gamma$, and $\delta$. However, the polarization units can generally have different coefficients. As an example, in Fig. \ref{fig:simplifiedNBPC}, the factor graph of a length-$8$ NBPC with flexible kernel coefficients is demonstrated. Considering flexible coefficients might result in improved performance of an NBPC. The optimization of kernel coefficients is typically aimed at accelerating the speed of polarization such that the difference between the quality of ``good" and ``bad" channels, synthesized by a polarization transform, is maximized \cite{savin}. We refer to the flexible kernel-coefficient NBPCs as generalized NBPCs (G-NBPCs). Our goal in this section is to propose codes that have very low decoding complexity under our proposed fast decoding algorithm.

There is a one-to-one correspondence between the factor graph of a polar code with its binary tree representation. For example, in the factor graph of Fig. \ref{fig:simplifiedNBPC}, the gray boxes are associated with different nodes in the binary tree of the code. Let the kernel coefficients and generator matrix associated to node $(\nu,s)$ in the binary-tree of a G-NBPC be denoted by $\mu^{(\nu,s)}$, $\gamma^{(\nu,s)}$, $\delta^{(\nu,s)}$ and $\boldsymbol{\mathcal{G}}^{(\nu,s)}$, respectively. Then, we can derive the following recursive relation
\begin{equation}
\boldsymbol{\mathcal{G}}^{(\nu,s)} = \begin{bmatrix} \mu^{(\nu,s)}\boldsymbol{\mathcal{G}}^{(2\nu,s-1)}  & 0 \\ 
\gamma^{(\nu,s)} \boldsymbol{\mathcal{G}}^{(2\nu+1,s-1)}&  \delta^{(\nu,s)}\boldsymbol{\mathcal{G}}^{(2\nu+1,s-1)} \end{bmatrix},
\end{equation}
where $\boldsymbol{\mathcal{G}}^{(2\nu,s-1)}$ and $\boldsymbol{\mathcal{G}}^{(2\nu+1,s-1)}$ are the generator matrices corresponding to the child nodes of node $(\nu,s)$. Also, for $s=1$, we have
\begin{equation}
\boldsymbol{\mathcal{G}}^{(\nu,1)} = \begin{bmatrix} \mu^{(\nu,1)}  & 0 \\ 
\gamma^{(\nu,1)} &  \delta^{(\nu,1)} \end{bmatrix}.
\end{equation}
Using the above notations, the input sequence $\mathbf{u}=[u_0,\dots,u_{N-1}]^T$ is encoded into a G-NBPC codeword using the following equation
\begin{equation}
\mathbf{c}^T=\mathbf{u}^T\boldsymbol{\mathcal{G}}^{(0,n)},
\end{equation}
where $\boldsymbol{\mathcal{G}}^{(0,n)}$ is the generator matrix of G-NBPC. It is noted that for fixed kernel-coefficients, $\boldsymbol{\mathcal{G}}^{(0,n)}$ would be equal to $\mathbf{G}^{\bigotimes n}_2$ in \eqref{eq:polarencfix}.  

As it was explained in Section \ref{sec_emssc}, CN and VN updates at different stages of binary tree require a permutation stage corresponding to the polarization kernel coefficients. In addition, according to Section \ref{sec:fastdecoding}, decoding of the special nodes (with the exception of Rate-0 and Rate-1 nodes) requires a pre-permutation of LLR-vectors as well as a post-multiplication of estimated hard-symbol vectors. In order to avoid such operations, we propose to deploy the G-NBPC structure with the following condition on kernel coefficients:
\begin{equation}
\mu^{(\nu,s)}=1,\gamma^{(\nu,s)}=1,\delta^{(\nu,s)} = 1\:\:\:\: \text{for}\:\:\:\: 1 \leq s < s_0,
\end{equation}
where $s_0$ is a threshold that is determined empirically. In the proposed structure, the kernel coefficients are only optimized for $s_0\leq s \leq n$. Also, we limit the maximum size of the special nodes to $2^{(s_0-1)}$. As a result, the kernel coefficients in the identified special nodes would be equal to $1$ meaning that all the permutation stages in the fast decoding are removed. As an example, suppose the threshold $s_0$ in Fig. \ref{fig:simplifiedNBPC} is set to $s_0=2$. This means the optimization is only done for the coefficients at levels $s=2,3$ and all the coefficients at level $s=1$ are set to $1$.  
\section{Latency Analysis}
In this section, the decoding latency of the proposed fast-LNBSC decoder is analyzed. To facilitate the analysis, the following assumptions are made which are similar to existing assumptions made for the latency analysis of BPCs. 
\begin{itemize}
\item[1)] There is no resource limitation so all the parallelizable instructions can be carried out in one time step. 
\item[2)] Non-binary field operations on symbols (such as \eqref{NBupmsg} and \eqref{eq:mrepoutput}) are carried out instantaneously\footnote{The non-binary XOR can be performed by applying bit-wise binary XOR on binary representation of NB symbols.}. 
\item[3)] Addition/subtraction of real numbers is performed in one time step.
\item[4)] The permutation of elements in LLR-vectors are carried out instantaneously\footnote{Permutation matrices corresponding to different kernel coefficients as well as special nodes are fixed and do not depend on a specific codeword or noise realization. The only permutation matrix that is altered during the decoding operation is $\boldsymbol{\Pi}^+_{\mu\chi_k^{(2\nu,s-1)}}$ in \eqref{eq:gperm}. Therefore, it is generally reasonable to assume permutation of LLR vectors can be done instantaneously.}.
\item [5)] The CN operation in \eqref{NBcn} consumes $2$ time-steps which includes $1$ time-step for performing $q^2$ parallel additions and another time-step for $q$ parallel minimum-finding operations. 
\item[6)] The VN operation consumes $2$ time-steps which includes $1$ time-step for performing the addition in \eqref{NBvn} and $1$ time-step for executing \eqref{bias_rem}.  
\item[7)] If, in accordance with \eqref{symbLLR}, the input vector in \eqref{NB_harddec} contains at least one zero member, then the non-binary hard-decision is equivalent to finding the symbol corresponding to the zero-element and can be carried out instantaneously. Otherwise, it takes $1$ time-step to perform \eqref{NB_harddec}. While the former comprises the majority of cases including the LLR vectors at the leaf nodes, the latter might happen during the decoding of certain special nodes such as M-REP nodes.  
\end{itemize}
Based on the above assumptions, it takes $4N-4$ time steps for LNBSC algorithm to decode a single non-binary codeword. In the following, we present the decoding latency of different special nodes.

Decoding the Rate-0 and Rate-1 nodes does not take any time steps. In fact, the output of a Rate-0 node is known in advance, and decoding a Rate-1 node consists of a hard decision which according to assumption $7$ can be done instantaneously. The fast-LNBSC decoder for an M-REP node computes the summation and hard decision of \eqref{eq_hd_rep} in $2$ time-steps. Note that according to assumptions $2$ and $4$, execution of \eqref{eq_permutLLR} and \eqref{eq:mrepoutput} do not take any time steps. Regarding the M-SPC nodes, the latency is determined by the worst-case scenario where the parity equation is not satisfied. The M-SPC node decoder in the simplified version of Algorithm \ref{alg:nbSPC2} (gray lines) computes the hard-decision and parity-check operation concurrently with the generation of candidates for each position in $1$ time-step. In the next time step, it selects the symbol candidate with minimum associated LLR. As such, the latency of the simplified algorithm is $2$ time-steps. In the full version of Algorithm \ref{alg_nbSPC2}, hard-decision and parity-check operation, generation of candidates for each position, and finding the second-minimum LLR in each column of LLR matrix are performed simultaneously and in $1$ time-step. The next time-step, storing the minimum LLR and its associated symbol candidate, finding the minimum value among second-minimum LLRs, and generating symbol candidates for different positions except the one corresponding to the second-minimum LLR are performed in parallel. In the last time-step, the penalties corresponding to the two generated outputs are compared and the output with the minimum penalty is selected. As such, the latency of Algorithm \ref{alg_nbSPC2} is $3$ time-steps. In the following, we assume M-SPC nodes are decoded by executing all the lines in Algorithm \ref{alg:nbSPC2}.

For Type-I node, our proposed decoder consumes $2$ time steps as it is equivalent to decoding two parallel M-REP nodes. Type-II decoder consumes one time-step for performing the summation in \eqref{eq:T2summationllr} and $3$ time-steps for decoding the M-SPC subnode. Hence, the latency of Type-II node is $4$ time-steps. The latency of Type-III node is similar to M-SPC node, i.e., $3$ time-steps. In Type-V nodes, $1$ time-step is consumed for summation in \eqref{eq:sumllr_T5}, $2$ time-steps is used for M-REP subnode followed by $3$ time-steps for M-SPC subnode. As such, the latency of Type-V node is $6$ time-steps.

For GM-REP nodes, computing the input LLRs of the source node consumes $1$ time-step. Therefore, denoting the latency of decoding the source node by $\Delta_s$, the total latency of a GM-REP node decoder is $1 + \Delta_s$.

In GM-PC nodes, there is no need for computing the input LLRs of the parity nodes that are of type Rate-0. However, if the parity node be a generic node with nonzero rate, depending on the implementation approach for non-binary CN operations, it might take $2\log_2(N_s/N_{s'})$ or $2$ time-steps to compute the input LLRs of the parity node\footnote{$N_s$ and $N_{s'}$ are the sizes of EG-PC node and its parity node, respectively.}. While the former corresponds to the low-complexity implementation based on basic CN units with two inputs, the latter is the required time steps when the CN operations on the left side of binary tree are merged and their output is calculated in one shot. In the following, we assume the CNs are merged and the input of the parity node is computed in $2$ time-steps. Note that the latency of decoding the parity node depends on its type. Also, $3$ time-steps are required for decoding $N_{s'}$ parallel M-SPC codes. As such, denoting the latency of the parity node by $\Delta_p$, the total latency of a GM-PC node is $3$ time-steps when the parity-node is Rate-0 and $5+\Delta_p$ time-steps when it is a generic node with nonzero rate. Type-IV node is an instance of GM-PC node in which the parity node is an M-REP node with $N_{s'}=4$ and $\Delta_p=2$, hence, the total latency of Type-IV node is $7$ time-steps.

\section{Simulation Results}
In this section, we compare the bit-error-rate (BER) and frame-error-rate (FER) of the proposed fast-LNBSC decoder with the LNBSC algorithm (without special nodes) as the baseline. To derive the results, systematic non-binary polar codes of length $N\in\{256,512,1024, 2048\}$ and rate $R\in\{0.25, 0.5,0.75\}$ are modulated using BPSK and transmitted over AWGN channel. The SNR in the presented results is considered as $1/\sigma^2$. Furthermore, the polar codes are constructed in $\mathbb{GF}(16)$ using MC method, \cite{nbplowlatency,savin}, where $f(x)=x^4+x+1$ is adopted as the primitive polynomial. The design SNR for different codes is set to values around which the FER is $10^{-3}$. We use quantized LNBSC decoder with EMS CN approximation, \eqref{NBcn}, in which the channel LLRs are represented by $5$ bits and the internal messages of the decoder are represented by $6$ bits. 

Fig. \ref{fig:fer_alg1} compares the FER performances of the fast-LNBSC decoder with the LNBSC decoder. We consider four rate-0.5 codes whose lengths are $256, 512$, $1024$, and $2048$, respectively. The kernel coefficients for all codes are optimized as\footnote{Here we consider the same approach as in \cite{savin,emspolar,2x2NBkernel,nbplowlatency,scma}, where only $\gamma$ is optimized while the other two coefficients are set to $1$.} $\mu=1$, $\gamma=\alpha^4$ and $\delta = 1$. In these results, only the gray lines in Algorithm $\ref{alg:nbSPC2}$ are used to decode the  M-SPC nodes. It is observed that fast-LNBSC decoder demonstrates between $0.1$dB to $0.3$dB loss in FER which is mainly because of the simplification in Algorithm \ref{alg:nbSPC2}. 
\begin{figure}
\centering
\includegraphics[width=3.8in]{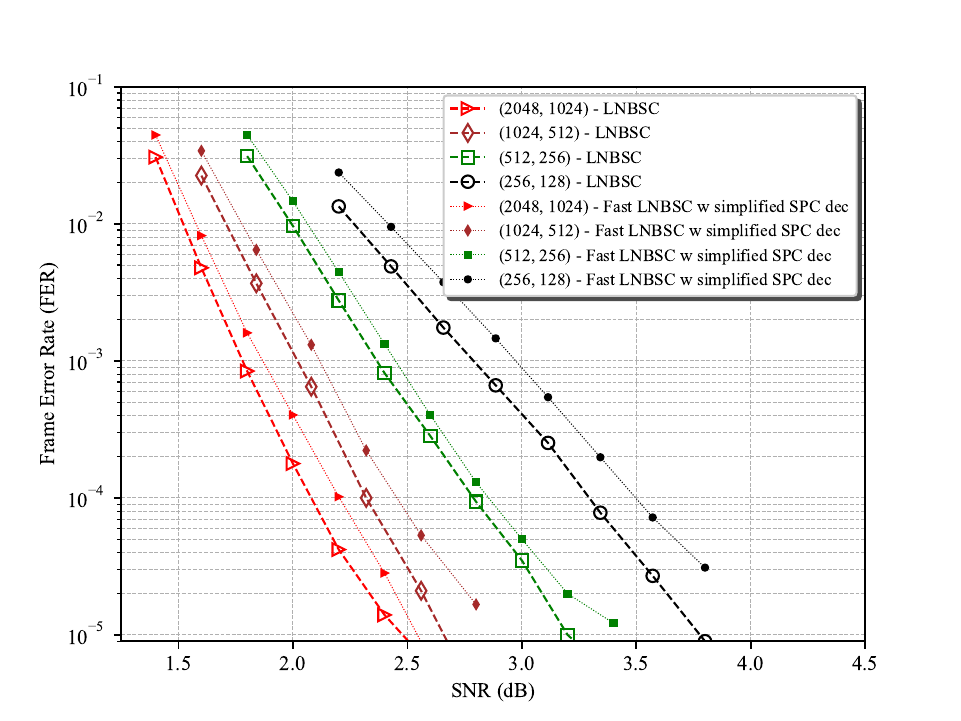}
\caption{FER performances of LNBSC and fast-LNBSC decoders. The simplified version of Algorithm \ref{alg:nbSPC2} is used to decode M-SPC nodes. Code lengths of $\{256, 512,1024, 2048\}$ with code-rate $0.5$ are considered.}
\label{fig:fer_alg1}
\end{figure} 

For the rest of the simulation results, we use fast-LNBSC decoder along with the complete version of Algorithm \ref{alg:nbSPC2} for decoding M-SPC nodes.

In figures \ref{fig:fer_3rate} and \ref{fig:ber_3rate}, the FER and systematic BER of polar codes with different lengths and rates are illustrated. $\gamma=\alpha^4$, $\mu=1$, and $\delta=1$ are considered as the kernel coefficients of all the simulated codes.  Note that the performance of our proposed fast-LNBSC decoder is identical to that of the LNBSC decoder with significantly lower latency. These results suggest that while the proposed method in Algorithm \ref{alg:nbSPC2} is not an ML decoder, it can achieve the performance of ML decoding with lower complexity. Table \ref{tab:nodDist} reports the required time steps and distribution of various special nodes for different code lengths and rates. In Table \ref{tab:nodDist}, GM-REP$^\dag$ denotes a subset of GM-REP nodes that does not include Type-I, Type-II, and Type-V nodes. Also, GM-PC$^\dag$ is a subset of GM-PC nodes in which the Type-III and Type-IV nodes are excluded. Note that the number of GM-REP$^\dag$ in all the cases is zero meaning that Type-I, Type-II, and Type-V nodes form all the corresponding GM-REP nodes. Moreover, the majority of GM-PC nodes are Type-III and Type-IV nodes, and, only, by increasing the code lengths, a few GM-PC$^\dag$ nodes are emerged. The required time steps for decoding each code using our proposed fast LNBSC decoder is illustrated in the last column of Table \ref{tab:nodDist}. Compared to the latency of LNBSC decoding, i.e., $4N-4$, the latency of our proposed fast decoder for codes with lengths $256, 512, 1024$ and $2048$, on average, are reduced by $91\%, 92\%, 94\%$ and $95\%$, respectively.
\begin{figure}
\centering
\includegraphics[width=3.8in]{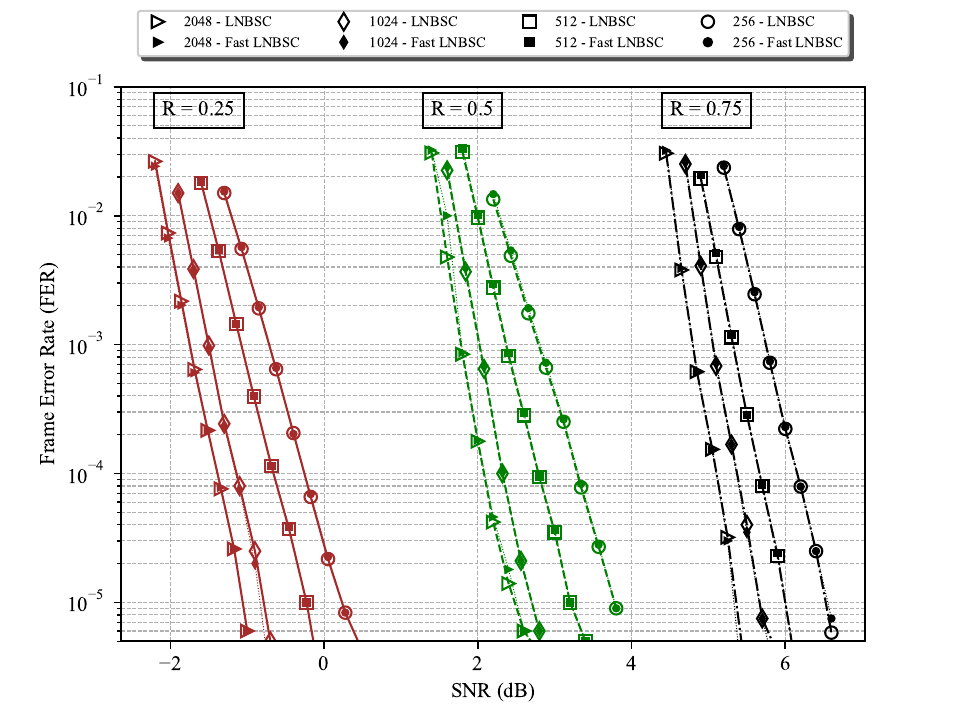}
\caption{FER performances of LNBSC and fast-LNBSC decoders. Algorithm \ref{alg:nbSPC2} is used to decode M-SPC nodes. Code-lengths of $\{256, 512,1024, 2048\}$ with code-rates $\{0.25,0.5,0.75\}$ are considered.}
\label{fig:fer_3rate}
\end{figure}
\begin{figure}
\centering
\includegraphics[width=3.8in]{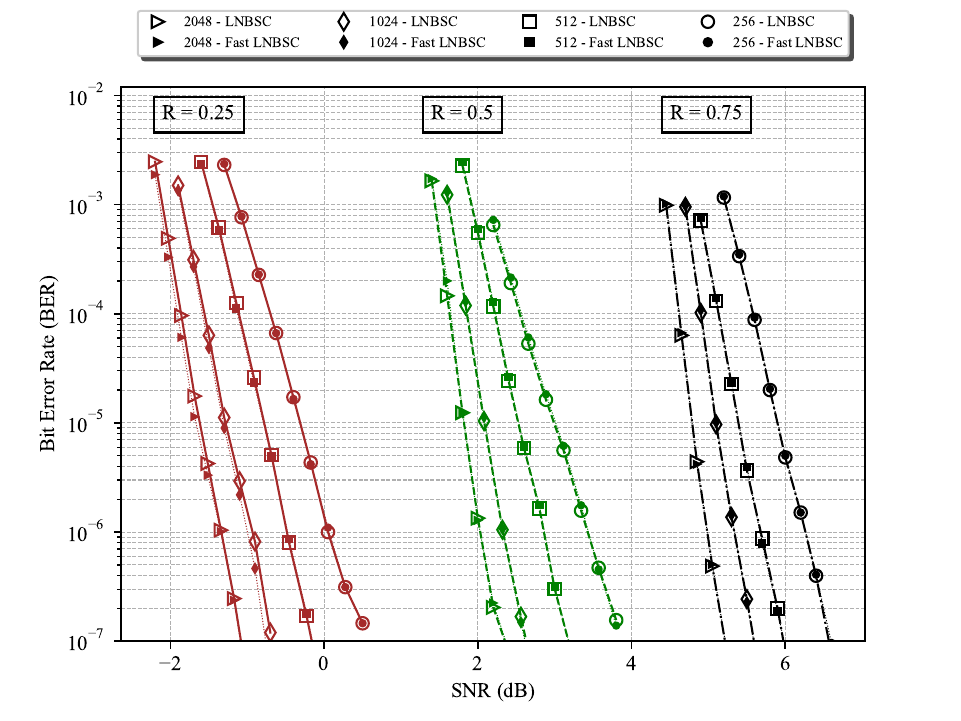}
\caption{Systematic BER performances of LNBSC and fast-LNBSC decoders. Algorithm \ref{alg:nbSPC2} is used to decode M-SPC nodes. Code-lengths of $\{256, 512,1024, 2048\}$ with code-rates $\{0.25,0.5,0.75\}$ are considered.}
\label{fig:ber_3rate}
\end{figure}

Figure \ref{fig:fer_varcoeff} compares the FER performances of the codes designed based on fixed kernel coefficients, $\gamma=\alpha^4$, $\mu=1$ and $\delta=1$, with the ones that are designed according to variable kernel coefficients. We use LNBSC to decode the codes with fixed kernel and fast-LNBSC for decoding the variable-kernel codes. For these results, three length-$1024$ codes with rates $\in\{0.25, 0.5, 0.75\}$ are considered. In the variable kernel cases, we use our proposed simplified structure in Section \ref{sec:Gnbpc} in which the threshold value is set to $s_0=8$ and the maximum size of special nodes, excluding Rate-0 and Rate-1 nodes, are limited to $2^7=128$ symbols. We applied the method of \cite{savin} to optimize the kernel coefficients such that the speed of polarization is maximized. Note that, likewise \cite{savin}, $\mu^{(\nu,s)}$ and $\delta^{(\nu,s)}$ for all the stages are set to $1$. For the codes with rates $0.25$, $0.5$ and $0.75$, the coefficients $\{\gamma^{(0,10)}, \gamma^{(0,9)}, \gamma^{(1,9)}, \gamma^{(0,8)}, \gamma^{(1,8)}, \gamma^{(2,8)}, \gamma^{(3,8)}\}$  were optimized as $\{\alpha^{10},\alpha^{11},\alpha^{9},1,\alpha^{7},\alpha^{11},\alpha^{11}\}$, $\{\alpha^{5},\alpha^{4},\alpha^{8},\alpha^{8},\alpha^{7},\alpha^{4},\alpha^{11}\}$ and  $\{\alpha^{10},\alpha^{4},\alpha^{8},\alpha^{8},\alpha^{8},\alpha^{11},\alpha^{4}\}$. It can be observed in Fig. \ref{fig:fer_varcoeff} that the FER performances of fixed-kernel and variable-kernel codes are pretty close. This is while by setting $\gamma^{(\nu,s) }= 1$ for $1\leq s < 8$, the fast decoding of the proposed variable-kernel code has been significantly simplified. For the codes considered in Fig. \ref{fig:fer_varcoeff}, Table \ref{tab:nodDist_vargamma} reports the latency and distribution of special nodes. Despite limiting the size of special nodes to less than $256$ symbols in variable-kernel codes, the required decoding time steps in both cases, i.e., variable and fixed-kernel codes, are relatively close. 
\begin{figure}
\centering
\includegraphics[width=3.8in]{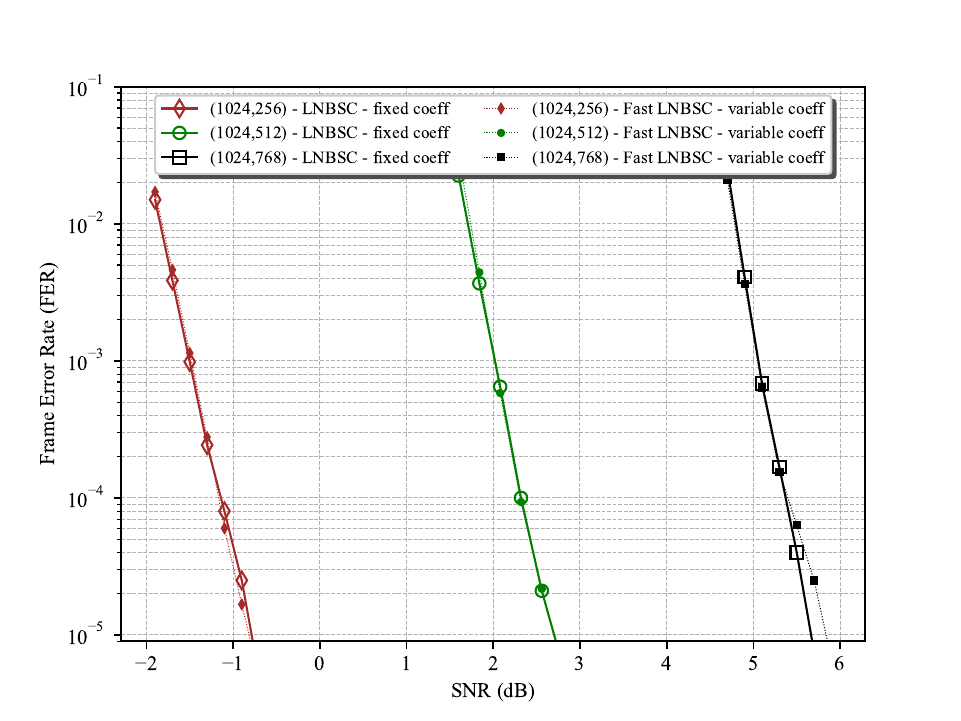}
\caption{FER performances of LNBSC and fast-LNBSC decoders. Algorithm \ref{alg:nbSPC2} is used to decode M-SPC nodes. Three codes with rates $\in\{ 0.25,0.5,0.75 \}$ and equal code lengths of $1024$ are considered. Fixed and variable kernel coefficients are used for code design.}
\label{fig:fer_varcoeff}
\end{figure}

\begin{table}[ht]\small
\caption{Distribution of Special Nodes for NBPCs Designed in $\mathbb{GF}(16)$}
\begin{center}
\scalebox{0.75}{
\begin{tabular}{|c|c||c|c|c|c|c|c|c|c|c|c|c||c|}
	\hline
	\begin{turn}{-90}{\textbf {Code-Rate}}\end{turn}& \begin{turn}{-90}{\textbf {Code-Length}$\:\:$}\end{turn}& \begin{turn}{-90}{\textbf {Rate-0}}\end{turn}& \begin{turn}{-90}{\textbf {Rate-1}}\end{turn}& \begin{turn}{-90}{\textbf {M-REP}}\end{turn}& \begin{turn}{-90}{\textbf {M-SPC}}\end{turn}& \begin{turn}{-90}{\textbf {Type-I}}\end{turn}& \begin{turn}{-90}{\textbf {Type-II}}\end{turn}& \begin{turn}{-90}{\textbf {Type-III}}\end{turn}& \begin{turn}{-90}{\textbf {Type-IV}}\end{turn}& \begin{turn}{-90}{\textbf {Type-V}}\end{turn}& \begin{turn}{-90}{\textbf {GM-REP}$^\dag$}\end{turn}& \begin{turn}{-90}{\textbf {GM-PC}$^\dag$}\end{turn}& \begin{turn}{-90}{\textbf {Time-Steps}}\end{turn}\\
	\hline
	\hline
	\multirow{4}{*}{$\mathbf{0.25}$}& $\mathbf{2048}$& $9$&$1$  & $11$& $13$& $1$& $1$& $2$& $1$& $12$& -& $1$& $350$\\
	\cline{2-14}
	& $\mathbf{1024}$& $4$&-  & $6$& $11$& -& $3$& $2$& -& $6$& -& -& $215$\\
	\cline{2-14}
	& $\mathbf{512}$& $4$& $1$& $6$& $6$& -& 1& -& $1$& $4$& -& -& $145$ \\
	\cline{2-14}
	& $\mathbf{256}$& -& -& $2$& $2$& $1$& -& $1$& -& $4$& -& -& $76$ \\
	\hline
	\hline
		\multirow{4}{*}{$\mathbf{0.5}$}& $\mathbf{2048}$& $6$& -& $12$& $18$& $4$& $2$& $3$& $2$& $10$& -& $3$& $438$\\
	\cline{2-14}
		& $\mathbf{1024}$& $5$&-  & $6$& $14$& $2$& $2$& $1$& -& $7$& -& $1$& $262$\\
	\cline{2-14}
	& $\mathbf{512}$& $2$& $2$& $4$& $10$& $1$& $1$& -& $1$& $5$& -& -& $178$ \\
	\cline{2-14}
	& $\mathbf{256}$& -& $1$& $2$& $3$& -& $1$& $1$& -& $4$& -& -& $88$ \\
	\hline
	\hline
		\multirow{4}{*}{$\mathbf{0.75}$}& $\mathbf{2048}$& $3$& $2$& $8$& $23$& $3$& $3$& $1$& $2$& $9$& -& $1$& $398$\\
	\cline{2-14}
		& $\mathbf{1024}$& $2$& $1$& $7$& $13$& $1$& -& $2$& -& $9$& -& $1$& $263$\\
	\cline{2-14}
	& $\mathbf{512}$& -& -& $2$& $10$& -& $2$& -& $2$& $4$& -& -& $156$ \\
	\cline{2-14}
	& $\mathbf{256}$& -& $1$& $2$& $6$& $1$& -& $1$& -& $3$& -& -& $98$ \\
	\hline
\end{tabular}}\\
\tiny %\scriptsize
\textit{GM-REP$^\dag$: a subset of GM-REP nodes that does not include Type-I, Type-II and Type-V nodes.$\:\:\:\:\:\:\:\:\:\:\:\:\:$}\\
\textit{GM-PC$^\dag$: a subset of GM-PC nodes that does not include Type-III and Type-IV nodes.$\:\:\:\:\:\:\:\:\:\:\:\:\:\:\: \:\:\:\:\:\:\:\:\:$}
\end{center}
\label{tab:nodDist}
\end{table}

\begin{table}[ht]\small
\caption{Distribution of Special Nodes for Length-$1024$ NBPCs Designed in $\mathbb{GF}(16)$}
\begin{center}
\scalebox{0.75}{
\begin{tabular}{|c|c||c|c|c|c|c|c|c|c|c|c|c||c|}
	\hline
	\begin{turn}{-90}{\textbf {Kernel}}\end{turn}
	& \begin{turn}{-90}{\textbf {Code-Rate}$\:\:$}\end{turn}& \begin{turn}{-90}{\textbf {Rate-0}}\end{turn}& \begin{turn}{-90}{\textbf {Rate-1}}\end{turn}& \begin{turn}{-90}{\textbf {M-REP}}\end{turn}& \begin{turn}{-90}{\textbf {M-SPC}}\end{turn}& \begin{turn}{-90}{\textbf {Type-I}}\end{turn}& \begin{turn}{-90}{\textbf {Type-II}}\end{turn}& \begin{turn}{-90}{\textbf {Type-III}}\end{turn}& \begin{turn}{-90}{\textbf {Type-IV}}\end{turn}& \begin{turn}{-90}{\textbf {Type-V}}\end{turn}& \begin{turn}{-90}{\textbf {GM-REP}$^\dag$}\end{turn}& \begin{turn}{-90}{\textbf {GM-PC}$^\dag$}\end{turn}& \begin{turn}{-90}{\textbf {Time-Steps$\:\:\:$}}\end{turn}\\
	\hline
	\hline
	\multirow{3}{*}{\textbf{Variable}}& $\mathbf{0.25}$& $4$&-  & $6$& $10$& -& $3$& $1$& $1$& $5$& -& $1$& $217$\\
	\cline{2-14}
	& $\mathbf{0.5}$& $6$& $1$& $7$& $14$& $1$& $3$& $2$& -& $8$& -& -& $277$ \\
	\cline{2-14}
	& $\mathbf{0.75}$& $2$& $1$& $8$& $13$& -& $1$& $1$& $1$& $8$& -& $1$& $261$ \\
	\hline
	\hline
		\multirow{3}{*}{\textbf{Fixed}}& $\mathbf{0.25}$& $4$&-  & $6$& $11$& -& $3$& $2$& -& $6$& -& -& $215$\\
	\cline{2-14}
	& $\mathbf{0.5}$& $5$&-  & $6$& $14$& $2$& $2$& $1$& -& $7$& -& $1$& $262$\\
	\cline{2-14}
	& $\mathbf{0.75}$& $2$& $1$& $7$& $13$& $1$& -& $2$& -& $9$& -& $1$& $263$ \\
	\hline

	\hline
\end{tabular}}\\
\tiny %\scriptsize
\textit{GM-REP$^\dag$: a subset of GM-REP nodes that does not include Type-I, Type-II and Type-V nodes.$\:\:\:\:\:\:\:\:\:\: \:\:\:\:\:\:\:\:\: $}\\
\textit{GM-PC$^\dag$: a subset of GM-PC nodes that does not include Type-III and Type-IV nodes.$\:\:\:\:\:\:\:\:\:\:\:\:\:\:\: \:\:\:\:\:\:\:\:\: \:\:\:\:\:\: $}
\end{center}
\label{tab:nodDist_vargamma}
\end{table}

\section{Conclusion}
In this work, a fast successive-cancellation (SC) decoding algorithm for non-binary polar codes (NBPCs), which we refer to as fast LLR-domain non-binary SC (LNBSC) decoder, is proposed. The main underlying idea behind our proposed method is the identification and fast decoding of certain non-binary special nodes in the decoding tree of the NBPC. To the best of our knowledge, this is the first work in which the identification of a variety of non-binary special nodes along with their decoding algorithms is proposed. We also proposed a simplified NBPC structure that facilitates the procedure of our proposed non-binary fast SC decoding. Our results show that the proposed fast non-binary decoder achieves significant decoding speed improvement with respect to the original SC decoder. This is while a significant reduction in the decoding latency is observed without sacrificing the error-rate performance of the code.

%\section*{Acknowledgments}
%This should be a simple paragraph before the References to thank those individuals and institutions who have supported your work on this article.

%{\appendix[Proof of the Zonklar Equations]
%Use $\backslash${\tt{appendix}} if you have a single appendix:
%Do not use $\backslash${\tt{section}} anymore after $\backslash${\tt{appendix}}, only $\backslash${\tt{section*}}.
%If you have multiple appendixes use $\backslash${\tt{appendices}} then use $\backslash${\tt{section}} to start each appendix.
%You must declare a $\backslash${\tt{section}} before using any $\backslash${\tt{subsection}} or using $\backslash${\tt{label}} ($\backslash${\tt{appendices}} by itself
% starts a section numbered zero.)}


\newpage

% your own custom macro containing the $\backslash${\tt{includegraphics}} command to make things
%\bf{If you include a photo:}\vspace{-33pt}
%\begin{IEEEbiography}[{\includegraphics[width=1in,height=1.25in,clip,keepaspectratio]{fig1}}]{Michael Shell}
%Use $\backslash${\tt{begin\{IEEEbiography\}}} and then for the 1st argument use $\backslash${\tt{includegraphics}} to declare and link the author photo.
%Use the author name as the 3rd argument followed by the biography text.
%\end{IEEEbiography}
%\vfill

\end{document}